\documentclass[11pt]{article}

\usepackage[margin=1.08in]{geometry}
\usepackage{amsmath, amssymb, amsthm}
\usepackage{mathtools}
\usepackage{hyperref}
\usepackage[capitalize,noabbrev]{cleveref}
\usepackage{microtype}
\usepackage{xcolor}
\usepackage{quantikz}
\usepackage{caption}
\usepackage{booktabs}
\usepackage{bbm}

\usepackage{thmtools}
\usepackage{thm-restate}

\DeclarePairedDelimiter{\norm}{\lVert}{\rVert}

\declaretheorem[name=Theorem]{theorem}
\declaretheorem[name=Proposition]{proposition}

\newtheorem{lemma}{Lemma}
\newtheorem{corollary}{Corollary}
\newtheorem{claim}{Claim}
\newtheorem{definition}{Definition}
\newtheorem{remark}{Remark}

\newcommand{\BQP}{\mathsf{BQP}}
\newcommand{\BPP}{\mathsf{BPP}}
\newcommand{\PH}{\mathsf{PH}}
\newcommand{\IQP}{\mathsf{IQP}}
\newcommand{\DQC}{\mathsf{DQC}}
\newcommand{\hBQP}{\tfrac{1}{2}\text{-}\mathsf{BQP}}
\newcommand{\Forr}{\Phi}
\renewcommand{\ket}[1]{|#1\rangle}
\renewcommand{\bra}[1]{\langle #1|}

\newcommand{\ha}{\widehat{a}}
\newcommand{\hb}{\widehat{b}}
\newcommand{\kb}[1]{|#1\rangle\!\langle #1|}
\newcommand{\one}{\mathbbm{1}}

\newcommand{\hsigma}{\widehat{\sigma}}

\newcommand{\zo}{\{0,1\}}

\newcommand{\odd}{\text{odd}}
\newcommand{\even}{\text{even}}
\newcommand{\COMMENT}[1]{}
\newcommand{\hf}{\widehat{f}}

\usepackage{titling}
\begin{document}
	
	\title{$\IQP$ circuits for $2$-\textsc{Forrelation}}
	\author{Quentin Buzet and Andr\'e Chailloux  
		\vspace{0.2cm} \\
		{\small Inria de Paris, COSMIQ team} \\
		{\small \texttt{quentin.buzet@inria.fr}} \quad {\small \texttt{andre.chailloux@inria.fr}}
	}
	\date{}

	\maketitle
	\vspace*{-1cm}
	\begin{abstract}
		The $2$-\textsc{Forrelation} problem provides an optimal separation between classical and quantum query complexity and is also the problem used for separating $\BQP$ and $\PH$ relative to an oracle.
		A natural question is therefore to ask what are the minimal quantum resources needed to solve this problem.
		We show that $2$-\textsc{Forrelation} can be solved using Instantaneous Quantum Polynomial-time ($\IQP$) circuits, a restricted model of quantum computation in which all gates commute.
		Concretely, signed $2$-\textsc{Forrelation} can be solved by a classical random choice between two one-query IQP circuits, while the absolute-value variant uses two independent executions of this randomized procedure.
		This answers a recent open question of Girish~\cite{Gir26} on the power of commuting quantum computations.
		For the Raz--Tal distribution, this randomization is unnecessary.
		We use this to show that there is an oracle $O$ such that $\IQP^O \not\subseteq \PH^O$, strengthening the result of Raz and Tal~\cite{RT22}.
		It also yields an oracle separation between $\IQP$ and $\DQC_1$.
		We prove Fourier growth bounds for multi-query IQP circuits, including bounds in terms of the size of their accepting set.
		Our results suggest a possible route toward decision-based quantum advantage within the restricted IQP model.	
		The key ingredient is an algebraic identity of the quadratic function $Q(x) = \sum_{i < j} x_ix_j$ that allows extracting inner-product phases within an $\IQP$ circuit.
	\end{abstract}
	

	\newcommand{\mo}{\{-1,1\}}
\newcommand{\negl}{{negl}}
\newcommand{\Image}{\mathrm{Im}}

\section{Introduction}
\label{sec:intro}

\subsection{Context}
\label{ssec:context}

One of the central questions in quantum complexity theory is to understand where the boundary between classical and quantum computation lies. There have been many results showing exponential quantum speedups for query complexity (\cite{DJ92,Sim97,BV93} for instance). But the full extent of possible quantum speedups, and especially which quantum resources are needed to achieve them, is still not well understood. The $2$-\textsc{Forrelation} problem, introduced by Aaronson~\cite{Aar10} and developed by Aaronson and Ambainis~\cite{AA15}, has emerged as a key problem for these questions.

\paragraph{The $2$-\textsc{Forrelation} problem.}
Given quantum oracle access to two functions $f,g:\{0,1\}^n\to\{-1,+1\}$, the task is to decide whether the \emph{forrelation}
\[
\Forr(f,g) = \frac{1}{2^{3n/2}}
\sum_{x,y \in \zo^n}
(-1)^{x \cdot y} f(x) g(y)
\]
is large or small in absolute value. The problem can also be defined with the promise that $\Phi(f,g)$ is large or small, this is the signed variant. A simple quantum circuit solves this problem in absolute value with a single query to $O_f$ and a single query to $O_g$ (Figure~\ref{Figure:1}).
\begin{figure}[h!]
	\centering
	\begin{quantikz}
		\lstick{$\ket{0^n}$} & \qwbundle{n} &
		\gate{H^{\otimes n}} &
		\gate{O_f} &
		\gate{H^{\otimes n}} &
		\gate{O_g} &
		\gate{H^{\otimes n}} &
		\meter{}
	\end{quantikz} 
	\caption{$\text{BQP}$ circuit for $2$-\textsc{Forrelation} with one query to $O_f$ and one query to $O_g$}
	\label{Figure:1}
\end{figure}

A direct calculation shows that this circuit outputs $0^n$ with probability $\Phi^2(f,g)$. Aaronson and Ambainis~\cite{AA15} also observed that $2$-\textsc{Forrelation} can be solved with a \emph{single query} to a joint oracle
$$ O_{f,g} \ket{b}\ket{x} = \begin{cases}
	f(x)\ket{0}\ket{x} & \text{if } b = 0, \\
	g(x)\ket{1}\ket{x} & \text{if } b = 1,
\end{cases}
$$
with the following $n+1$ qubit circuit below.

\begin{figure}[h!]
	\centering
	\begin{quantikz}
		\lstick{$\ket{0}$} & \gate{H}  & \gate[2]{O_{f,g}} & \ctrl{1}& \gate{H} & \meter{} \\
		\lstick{$\ket{0^n}$} & \gate{H^{\otimes n}} & & \gate{H^{\otimes n}} &  & \qw \\
	\end{quantikz}
	\caption{$\text{BQP}$ circuit for $2$-\textsc{Forrelation} with a single query to $O_{f,g}$}
	\label{Figure:2}
\end{figure}

This circuit measures $0$ on the first qubit with probability $\frac{1}{2} + \frac{1}{2}\Phi(f,g)$, giving a constant bias for the signed variant.

Regarding the classical hardness of $2$-\textsc{Forrelation}, Aaronson and Ambainis~\cite{AA15} proved a nearly tight classical lower bound of $\Omega(\frac{2^{n/2}}{\log(n)})$ randomized queries. Bansal and Sinha~\cite{BS21} extended the tight separation to $k$-Forrelation for all $k$. Sherstov, Storozhenko, and Wu~\cite{SSW21}, via different techniques, also obtained the tight separation for all $k$ and moreover removed the logarithmic factor for $2$-\textsc{Forrelation}, yielding the tight $\Theta(2^{n/2})$ bound for randomized query complexity and confirming that $2$-\textsc{Forrelation} achieves the maximal possible quantum vs. classical query separation. Beyond the query model, Raz and Tal~\cite{RT22} used $2$-\textsc{Forrelation} to give the first oracle separation between $\BQP$ and the polynomial hierarchy~$\PH$. 

Despite achieving optimal quantum query advantage, and surpassing the power of $\PH$, $2$-\textsc{Forrelation} is not known, nor believed, to be complete for $\BQP$\footnote{$k$-\textsc{Forrelation} is known to be $\BQP$-complete for $k=\text{poly}(n)$ \cite{AA15}.}. This raises a natural question: \emph{which restricted quantum models can still solve $2$-\textsc{Forrelation}?} Understanding this would help map out the intermediate landscape between classical and fully quantum computation. Several recent works have introduced complexity classes that sit between $\BPP$ and $\BQP$, capturing the power of noisy or resource-limited quantum circuits.

\paragraph{Complexity classes below $\BQP$ and Fourier growth.}

Jacobs and Mehraban~\cite{JM24} introduced the class $\hBQP$, in which the computation starts from half of a maximally entangled state; the other half is revealed only at the end and measured in the computational basis. Another family are $\DQC_k$ circuits~\cite{KL98}, where the computation starts with $k$ clean qubits (in the $\ket{0}$ state) and the remaining qubits are in the maximally mixed state, as depicted in Figure~\ref{Figure:hBQP}.
\begin{figure}[h!]
	\centering
	\begin{quantikz}
		\lstick[wires=2]{$\dfrac{1}{\sqrt{2^n}}
			\displaystyle\sum_{x \in \zo^n}\ket{x}\ket{x}$}
		& \qwbundle{n} & \gate[wires=1]{U} & \meter{} \\
		& \qwbundle{n} & \qw & \meter{} 
	\end{quantikz} $\quad ; \quad$
	\begin{quantikz}
		\lstick{$\ket{0^k}$}  & \gate[wires=2]{U}
		& \meter{} \\
		\lstick{$\frac{I}{2^n}$}
		& \qw &  
	\end{quantikz}
	\caption{On the left: generic $\hBQP$ Circuit. On the right: generic $\DQC_k$ circuit.}
	\label{Figure:hBQP}
\end{figure}

Even $\DQC_1$, with just one clean qubit, can solve problems that are believed to be classically intractable, such as estimating the normalized trace of a unitary~\cite{KL98} or computing Jones polynomials~\cite{SJ08}. Moreover, any $\DQC_k$ circuit with $k = O(\log n)$ can be simulated by a $\hBQP$ circuit~\cite{JM24}, so $\hBQP$ is the more powerful model. The $2$-\textsc{Forrelation} problem separates these circuit families: it lies within $\hBQP$~\cite{JM24} but requires an exponential number of queries for $\DQC_k$ circuits when $k = O(\log(2^n))$ \cite{Gir26}.

The $\hBQP$ upper bound comes from an adaptation of the single-query circuit of~\cite{AA15}. The $\DQC_k$ lower bound uses a technique called \emph{Fourier growth}~\cite{Gir26,BS21,IRR+21}, which bounds the low-weight Fourier coefficients of the acceptance probability of a quantum circuit and thereby places limits on which problems a given model can solve. 

\paragraph{IQP circuits.} In this work, we will focus on another family of quantum circuits which are weaker than $\BQP$ circuits. Introduced by Shepherd and Bremner~\cite{SB09}, \emph{Instantaneous Quantum Polynomial-time} ($\IQP$) circuits on $m$ qubits are quantum circuits, where each gate is diagonal in the $X$ basis. Since each pair of gates commute, they can be applied in any order (or simultaneously), hence the name instantaneous. An equivalent formulation of $\IQP$ circuits are circuits on $m$ qubits of the form $C = H^{\otimes m}\, D\, H^{\otimes m}$, where $D$ is a diagonal unitary in the computational basis (i.e $Z$ basis). The circuit starts on $\ket{0^m} $ and is measured in the computational basis (Figure~\ref{Figure:IQP}).
When using an $\IQP$ circuit as a decision procedure, we allow acceptance to be specified by an efficiently decidable set $F$ of computational-basis measurement outcomes.
We formalize this convention in Section~\ref{Section:2}.

\begin{figure}[h!]
	\centering
	\begin{quantikz}
		\lstick{$\ket{0^m}$} & 
		\gate{H^{\otimes m}} & \gate{D} &
		\gate{H^{\otimes m}} & \meter{}
	\end{quantikz}
	\caption{Generic $\IQP$ circuit}
	\label{Figure:IQP}
\end{figure}

$\IQP$ circuits are structurally much easier to implement than regular $\BQP$ circuits. The fact that all gates can be applied simultaneously implies that the total coherence time of the circuit can be much lower. Error correction is also simpler as the only gates we apply in the central layer are diagonal in the $Z$ basis. Despite its restricted structure, $\IQP$ is believed to be strictly more powerful than $\BPP$: Bremner, Jozsa, and Shepherd~\cite{BJS11} showed that $\BPP$ cannot sample from the output distribution of $\IQP$ circuits unless $\PH$ collapses to its third level, and Bremner, Montanaro, and Shepherd~\cite{BMS16} strengthened this to an average-case hardness result. Jacobs and Mehraban showed that any $\IQP$ circuit can be simulated in $\hBQP$~\cite{JM24}, which bounds $\IQP$ from above. The relation between these three quantum circuit classes is depicted in Figure~\ref{Figure:BelowBQP}. In~\cite{Gir26}, Girish left the following open question:

\begin{figure}[h!]
	\centering
	\begin{tikzpicture}[font=\large]
		\draw[thick] (-5,-3.2) rectangle (5,3.2);
		\node at (-3.8,2.5) {$\mathrm{BQP}$};
		
		\draw[thick] (0,0) ellipse (4 and 2.5);
		\node at (-2,1.5) {$\tfrac{1}{2}\text{-}\mathrm{BQP}$};
		
		\draw[thick] (-1.0,0) circle (1.25);
		\node at (-1.5,0) {$\mathrm{DQC_1}$};
		
		\draw[thick] (1.0,0) circle (1.25);
		\node at (1.5,0) {$\mathrm{IQP}$};
		
		\begin{scope}[shift={(5.5,0)}]
			\node[anchor=west] at (0,1) {\textbf{}};
		\end{scope}
	\end{tikzpicture}
	\caption{Complexity classes below $\BQP$. The $2$-\textsc{Forrelation} problem can be solved in $\hBQP$ but not in $\DQC_1$ while the $3$-\textsc{Forrelation} problem can be solved in $\BQP$ but not in $\hBQP$ \cite{Gir26}.}
	\label{Figure:BelowBQP}
\end{figure}

\begin{quote}
	Is $2$-\textsc{Forrelation} solvable in
	$\IQP$ and if not, can we prove Fourier growth bounds?
\end{quote}

In the oracle/query model, the oracles $O_f$ and $O_g$ (or equivalently $O_{f,g}$) are phase oracles, hence diagonal in the computational basis. Thus, at the query level, they commute with the other diagonal operations and can be included in the diagonal block $D$ of an $\IQP$ circuit. Equivalently, they may be viewed as products of $Z$, $CZ$, $CCZ$, and higher-controlled $Z$-type gates, with each phase query counted as a single oracle gate. The difficulty is that the known quantum circuits for $2$-\textsc{Forrelation} are not of $\IQP$ form: they either place a Hadamard transform between $O_f$ and $O_g$, or use a controlled Hadamard (see Figures~\ref{Figure:1} and~\ref{Figure:2}). Indeed, if one tries to rewrite the standard one-query algorithm so that the final operation is a Hadamard layer, the remaining middle operation contains $\kb{0}{0}\otimes H^{\otimes n}+ \kb{1}{1}\otimes I^{\otimes n}$, which is not diagonal in the computational basis. Thus the standard algorithm is not simply an $\IQP$ circuit in disguise; the challenge is to reproduce the effect of this intermediate controlled Hadamard using only a diagonal block and classical pre/post-processing.

On the contrary, many well-known quantum query algorithms already are $\IQP$ circuits when the oracle is provided as a phase oracle: Bernstein-Vazirani~\cite{BV93}, Deutsch-Jozsa~\cite{DJ92}, Simon~\cite{Sim97}, hidden-shift for quadratic bent functions~\cite{Rot09}...

\subsection{Short Overview of contributions}
\label{Section:Contributions}

\subsubsection{$\IQP$ circuits for $2$-\textsc{Forrelation}.} 

The central contribution of this paper is to show that $2$-\textsc{Forrelation} can be solved using $\IQP$ circuits with only efficient classical pre-processing.

Concretely, for signed $2$-\textsc{Forrelation}, we construct two one-query $\IQP$ circuits on $n+1$ qubits and choose one of them uniformly at random.
The resulting acceptance probability equals $\frac{1}{2} + \frac{1}{2\sqrt{2}}\Phi(f,g)$ when $n$ is odd, and $\frac{1}{2} + \frac{1}{4}\Phi(f,g)$ when $n$ is even.
This randomized procedure makes a single call to an $\IQP$ circuit which performs a single quantum query to $O_{f,g}$.
The pre-processing consists of flipping a random coin, which determines the $\IQP$ circuit we run, and the circuit accepts according to an efficiently computable accepting set $F$.
For the absolute value variant, we also show how to solve $2$-\textsc{Forrelation}.
We simply run the above randomized procedure twice and accept if both runs accept or if both runs reject.
A simple calculation shows that this procedure accepts with probability $\frac{1}{2} + \frac{1}{4}\Phi^2(f,g)$ if $n$ is odd and $\frac{1}{2} + \frac{1}{8}\Phi^2(f,g)$ if $n$ is even. 

The key tool is the quadratic function $Q(x) = \sum_{i < j} x_ix_j$ for $x = (x_1,\dots,x_n) \in \zo^n$. It is a classical object in the study of quadratic forms over $\mathbb{F}_2$, closely related to the second-order Reed--Muller code. Quadratic functions have also been a central object of study in the context of $\IQP$ circuits~\cite{BJS11,BMS16}. The key identity is $Q(x) + Q(y) + Q(x+y) = x \cdot y + |x||y|\pmod{2}$, which can be understood as saying that $Q$ nearly bilinearizes the inner product, with a correction term $|x||y|$.

When $|x+y|$ is odd, at least one of $|x|$, $|y|$ is even, so $|x||y| \mod2$ disappears. When $n$ is odd, we show how to use the above identity to choose a diagonal operator $D$ and acceptance set $F$, which gives us an $\IQP$ circuit accepting with probability\ $\frac{1}{2} + \frac{1}{\sqrt{2}}\Phi_{\odd}(f,g)$, where $\Phi_{\odd}(f,g)$ is the quantity $\Phi(f,g)$ restricted to summing over pairs $(x,y)$ with $|x+y|$ odd. A symmetric construction handles $\Phi_{\even}$, and by performing one of the two circuits at random, we obtain an accepting probability that depends on  $\Phi(f,g)$. For $n$ even, a simple one-bit padding reduces to the odd case at a constant factor cost.

\subsubsection{Oracle separation between $\IQP$ and $\PH$.}
As a corollary, our result strengthens the oracle separation of Raz and Tal~\cite{RT22} between $\BQP$ and the polynomial hierarchy to the $\IQP$ setting.
The Raz--Tal separation is based on a distributional version of $2$-\textsc{Forrelation} that is efficiently distinguishable quantumly but hard for $\PH$.
This distribution has the crucial property that
$$\mathbb{E}_{(f,g)\leftarrow\mathcal{D}}[\Phi_{\mathrm{odd}}(f,g)]
=\mathbb{E}_{(f,g)\leftarrow\mathcal{D}}[\Phi_{\mathrm{even}}(f,g)]
=\frac{1}{2}\mathbb{E}_{(f,g)\leftarrow\mathcal{D}}[\Phi(f,g)]
=\Omega\left(\frac{1}{n}\right),
$$
whereas these expectations are zero under the uniform distribution.

For this specific distribution, either of our two fixed one-query $\IQP$ circuits already achieves inverse-polynomial distinguishing advantage, so the random choice required for general inputs is unnecessary.
By running polynomially many copies of this circuit in parallel, their tensor product remains a single $\IQP$ circuit; choosing an appropriate accepting set then gives constant distinguishing advantage.
Combining this with the hardness result of Raz and Tal yields an oracle $O$ such that
$\IQP^O \not\subseteq \PH^O.$
Since $\IQP$ is a restricted subclass of $\BQP$, this strengthens the original oracle separation of Raz and Tal.

Combining our $\IQP$ upper bound with the $\DQC_1$ query lower bound of Girish~\cite{Gir26} also yields an oracle $O$ such that
$$\IQP^O \not\subseteq \DQC_1^O.$$

\subsubsection{Fourier growth bounds for $\IQP$.}

The acceptance probability function $p$ of any $d$-query algorithm is a polynomial in the oracle bits, and the level-$\ell$ Fourier growth $L_{1,\ell}(p)$ measures how much weight this polynomial puts on degree-$\ell$ coefficients.
For some restriction $\rho$, writing $N=2^n$, solving $2$-\textsc{Forrelation} with constant advantage requires $L_{1,2}(p|_\rho) = \Omega(\sqrt{N})$~\cite{BS21,Gir26}, so upper bounds on $L_{1,2}(p|_\rho)$ translate into query lower bounds. Table~\ref{tab:fourier} summarizes the known bounds for different models and compares them with ours. We prove for $d$-query $\IQP$ circuits the bound $L_{1,2}(p|_\rho)=O(\sqrt{dN})$, together with the accepting-set bound $L_{1,2}(p|_\rho)\leq \sqrt{|F|}$ in terms of the accepting set. More generally, for every fixed level $\ell$, we obtain unsigned Fourier growth bounds of order $O_\ell(d^{\lceil \ell/2\rceil/2}N^{\lfloor \ell/2\rfloor/2})$. Our accepting-set bound implies that any $\IQP$ circuit solving $2$-\textsc{Forrelation} with constant advantage must have $|F| = \Omega(N)$, and our construction of Section~\ref{Section:3} achieves $|F| = \Theta(N)$.

\begin{table}[ht]
	\centering
	\begin{tabular}{lcc}
		\toprule
		Model & $L_{1,2}(p|_\rho)$ & Queries for $2$-\textsc{Forrelation} \\
		\midrule
		$\BQP$~\cite{AA15,IRR+21}    & $d^2\sqrt{N}$  & $O(1)$   \\[2pt]
		$\hBQP$~\cite{JM24,Gir26} & $d^2\sqrt{N}$    & $O(1)$ \\[2pt]
		$\DQC_1$~\cite{Gir26} & $d^2$             & $d = \Omega\left(N^{1/4}\right)$ \\[2pt]
		$\boldsymbol{\IQP}$ \cite{GSTW24},\textbf{[this work]} & $\boldsymbol{O\!\left(\min\left\{\sqrt{|F|},\sqrt{dN}\right\}\right)}$ & $\boldsymbol{O(1)}$ \\[2pt]
		$\BPP$~\cite{Tal20,SSW21} & $d\sqrt{n}$ & $d = \Omega\left(\sqrt{N}/\sqrt{n}\right)$ \\
		\bottomrule
	\end{tabular}
	\caption{Upper bounds on the level-$2$ Fourier growth
		of $d$-query algorithms, up to $O(1)$ factors, and the lower/upper bounds on the number of queries $d$ needed to solve $2$-\textsc{Forrelation} (which requires $L_{1,2}(p|_\rho) = \Omega(\sqrt{N})$).
	}
	\label{tab:fourier}
\end{table}

We present in Section~\ref{Section:Discussion} a discussion of the implications of these results, but we first provide the detailed overview of our contributions.

\subsection{Detailed overview of contributions and techniques used}

\subsubsection{Main construction}
\begin{restatable}{theorem}{TheoremMain}\label{Theorem:Main}
	Let functions $f,g : \zo^n \rightarrow \mo$.
    There exists a randomized procedure that chooses uniformly between two $\IQP$ circuits, each making a single quantum query to $O_{f,g}$, and accepts with probability
	\begin{align*}
		P_{acc} & = \frac{1}{2} + \frac{1}{2\sqrt{2}} \Phi(f,g) & \text{if } n \ \odd \\
		P_{acc} & = \frac{1}{2} + \frac{1}{4} \Phi(f,g) & \text{if } n \ \even 
	\end{align*}
\end{restatable}

The starting point is the observation that the oracle $O_{f,g}$ is diagonal in the computational basis, so it can be absorbed into the diagonal layer of an $\IQP$ circuit. The challenge is that known algorithms for $2$-\textsc{Forrelation} use the oracle between two layers of Hadamards (Figure~1), while an $\IQP$ circuit has only one Hadamard layer on each side.

Our circuit uses $n+1$ qubits and calls the oracle exactly once. Since $O_{f,g}$ commutes with any other diagonal operator, we can assume the oracle is applied first, followed by a diagonal $D$ independent of $f,g$. The final state is measured in the computational basis and the circuit accepts if the outcome lies in some set $F \subseteq \zo^{n+1}$. Such a circuit can be represented as follows.

\begin{figure}[!h]
	$$
	\begin{quantikz}
		\lstick{$\ket{0^{n+1}}$} & 
		\gate{H^{\otimes (n+1)}} & \gate{O_{f,g}} & \gate{D} & 
		\gate{H^{\otimes (n+1)}} & \gate{\Pi_F}
	\end{quantikz}
	$$
	\caption{$\IQP$ circuit on $n+1$ qubits with a single call to $O_{f,g}$ and with a projection on the accepting space}
	\label{Figure:6}
\end{figure}

A circuit of this form has two degrees of freedom: the choice of $D$ and the choice of $F$. We write $D = \sum_{x \in \zo^{n+1}} e^{i \varphi_x} \kb{x}$ with each $\varphi_x \in [0,2\pi)$.
A natural case where the accepting probability has a simple form is when 
\begin{itemize}
	\item $D = \sum_{b \in \zo} \sum_{x \in \zo^n} {\rho_b(x)}\kb{bx}$ for functions $\rho_0,\rho_1 : \zo^n \rightarrow \mo$.
	\item and $F = \{0\} \times T \cup \{1\} \times \overline{T}$, with $T \subseteq \zo^n$ and $\overline{T} = \zo^n \backslash T$.
\end{itemize} 
Here and throughout, for any function $h : \zo^n \rightarrow \mathbb{C}$, we define \\ $\widehat{h}(x) = \frac{1}{\sqrt{2^n}} \sum_{y \in \zo^n} (-1)^{x \cdot y}h(y)$. In the above setting, we show with a direct calculation using convolution relations that 

\begin{align*}P_{acc} = \frac{1}{2} + \sqrt{\frac{1}{2^{3n + 2}}}\sum_{x,y \in \zo^n} f(x)g(y) \rho_0(x)\rho_1(y)\sigma(x+y) \quad \text{where} \quad
	\sigma  = \widehat{(-1)^{\one_{\overline{T}}}},
\end{align*}

The above accepting probability can be achieved for any choice of $\rho_0,\rho_1 : \zo^n \rightarrow \mo$. The definition $\sigma = \widehat{(-1)^{\one_{\overline{T}}}}$, or equivalently $\hsigma = (-1)^{\one_{\overline{T}}}$ enforces $|\hsigma| = 1$ which is a bentness condition for $\sigma$. The following proposition gives the key conditions on $\rho_0, \rho_1, \sigma$.

\begin{restatable}{proposition}{PropositionMain}\label{Proposition:2}
	If there exist functions $\rho_0,\rho_1 : \zo^n \rightarrow \mo$ and a function $\sigma : \zo^n \rightarrow \mathbb{R}$ such that 
	\begin{itemize}
		\item The diagonal operator defined by $\rho_0,\rho_1$
		admits a polynomial-size implementation using computational-basis diagonal gates.
		\item $\Image(\hsigma) \subseteq \{-1,1\}$ (the bentness condition) and $\hsigma$ is efficiently computable by a classical algorithm.
	\end{itemize}
	Then there exists an $\IQP$ circuit of the form of Figure~\ref{Figure:6} that accepts with probability
	$$ P_{acc} = \frac{1}{2} + \sqrt{\frac{1}{2^{3n + 2}}}\sum_{x,y \in \zo^n} f(x)g(y) {\rho_0(x) \rho_1(y)}\sigma(x+y).$$
\end{restatable}

This accepting probability is derived using a direct calculation and using convolution relations.
By assumption, the diagonal operator $D$ has an efficient diagonal implementation. Computing efficiently $\hsigma$ means we have an efficient membership algorithm for $T$, and hence we can perform the projection $\Pi_F$ efficiently. The sum in the accepting probability resembles the quantity $\Phi(f,g)$. If we could choose functions $\rho_0,\rho_1$ and $\sigma$ such that
\begin{align}\forall x,y \in \zo^n, \ {\rho_0(x) \rho_1(y)}\sigma(x+y) = (-1)^{x \cdot y},
\end{align} 
or at least have a proportionality relation between these quantities, then we would be done. We prove however that no such functions exist. The impossibility proof does fortunately reveal that a weaker version of the relation can hold: if we restrict to pairs $(x,y)$ where $|x+y|$ has a fixed parity, then suitable functions do exist.

\paragraph{The quadratic function $Q$.} The key function that we will use is the quadratic function 
$$\forall x = (x_1,\dots,x_n) \in \zo^n, \ Q(x) = \sum_{1 \le i < j \le n} x_ix_j.$$ 
This function has actually a very simple form, 
$Q(x) = \binom{|x|}{2} \mod2$. The key relation is 
$$ Q(x) + Q(y) + Q(x+y) = x \cdot y + |x||y| \mod 2.$$
In particular, when $|x+y|$ is odd, at least one of $|x|$, $|y|$ is even, so $|x||y| \equiv 0 \mod2$ and $Q(x) + Q(y) + Q(x+y) = x \cdot y$. We choose  $\rho_0 = \rho_1 = (-1)^Q$. To enforce the odd parity of $|x+y|$, we define $\sigma(x) = \sqrt{2}(-1)^{Q(x)}$ for $|x|$ odd, $\sigma(x) = 0$ otherwise, where the $\sqrt{2}$ factor is required for the bentness condition. With this choice of functions, we have 
\begin{itemize}
	\item If $|x+y|$ is odd, $\rho_0(x)\rho_1(y)\sigma(x+y) = \sqrt{2} (-1)^{Q(x) + Q(y) + Q(x+y)} = \sqrt{2}(-1)^{x \cdot y}$.
	\item If $|x+y|$ is even, $\rho_0(x)\rho_1(y)\sigma(x+y) = 0$.
\end{itemize}

Moreover we can show that 

\begin{proposition}
	$\hsigma(x) = \sqrt{2}\sin(\frac{\pi}{4}(n - 2|x|))$. In particular, if $n$ is odd, $|\hsigma(x)| = 1$. 
\end{proposition} 

Finally, the functions $\rho_0,\rho_1,\hsigma$ are efficiently computable. Therefore, with this choice of functions, we can use Proposition~\ref{Proposition:2} to show that there exists an $\IQP$ circuit performing a single query to $O_{f,g}$ that accepts with probability

\begin{align*}
	P_{acc} & = \frac{1}{2} + \sqrt{\frac{1}{2^{3n + 2}}}\sum_{x,y \in \zo^n} f(x)g(y) {\rho_0(x) \rho_1(y)}\sigma(x+y) \\
	& = \frac{1}{2} + \frac{1}{\sqrt{2}}\Phi_{\odd}(f,g) \quad \text{with} \quad \Phi_{\odd}(f,g) = \sqrt{\frac{1}{2^{3n}}}\sum_{\substack{x,y \in \zo^n \\ |x+y| \ \odd}} (-1)^{x \cdot y} f(x)g(y).
\end{align*}

A different choice of functions $\rho_0,\rho_1,\sigma$ handles the case $|x+y|$ even, again for odd $n$, yielding

$$ P_{acc} = \frac{1}{2} + \frac{1}{\sqrt{2}}\Phi_{\even}(f,g) \quad \text{with} \quad \Phi_{\even}(f,g) = \sqrt{\frac{1}{2^{3n}}} \sum_{\substack{x,y \in \zo^n \\ |x+y| \ \even}} (-1)^{x \cdot y} f(x)g(y).$$

In order to conclude, we run one of the two $\IQP$ circuits presented above uniformly at random.
Since $\Phi(f,g) = \Phi_\odd(f,g) + \Phi_{\even}(f,g)$, we can average over these two circuits accepting probabilities to obtain a randomized procedure that accepts with probability $\frac{1}{2} + \frac{1}{2\sqrt{2}}\Phi(f,g)$.
This gives the proof of Theorem~\ref{Theorem:Main} when $n$ is odd.

\begin{remark}
	This is where we need classical pre-processing in order to flip the random coin.
	Since the accepting set $F = \{0\} \times T \cup \{1\} \times \overline{T}$ will depend on the randomness, we are not able to compile this in a single $\IQP$ circuit without pre-processing.
	It would be interesting whether there is a way to compile this construction into a single $\IQP$ circuit without this random coin.
    We note, however, that for the Raz--Tal distribution used in the oracle separation, Section~4 shows that either of the two fixed $\IQP$ circuits already achieves inverse-polynomial distinguishing advantage, so no random coin is needed there.
\end{remark}

\subsubsection{Other circuit constructions}

\paragraph{The $n$ even case.} The construction above handles odd $n$. For even $n$, we do not know of functions $\rho_0,\rho_1,\sigma$ satisfying the required conditions directly. Instead, we reduce to the odd case by a simple padding argument: given functions $f,g : \zo^{n} \rightarrow \mo$, we extend them to $\zo^{n+1}$ by adding a dummy variable that does not affect their output.
Since $n+1$ is odd, the previous construction applies.
This padding costs a constant factor in the bias, and a direct calculation shows that we get a randomized procedure using $\IQP$ circuits for the even case that accepts with probability

$$ P_{acc} = \frac{1}{2} + \frac{1}{4}\Phi(f,g).$$

\paragraph{Dealing with $|\Phi(f,g)|$.}
Many formulations of the $2$-\textsc{Forrelation} problem consider the quantity $|\Phi(f,g)|$ rather than $\Phi(f,g)$.
We handle this at the cost of doubling the number of queries.
For example, when $n$ is odd, the randomized procedure of Theorem~\ref{Theorem:Main} accepts with probability $\frac{1}{2} + \frac{1}{2\sqrt{2}}\Phi(f,g)$ using a single query to $O_{f,g}$.
The idea is simple: run two independent copies of the procedure and accept if both give the same outcome. The resulting acceptance probability is 
$$ P_{acc} = \left(\frac{1}{2} + \frac{1}{2\sqrt{2}}\Phi(f,g)\right)^2 + \left(\frac{1}{2} - \frac{1}{2\sqrt{2}}\Phi(f,g)\right)^2 = \frac{1}{2} + \frac{1}{4}\Phi^2(f,g).$$
We can also make the same construction for $n$ even, and we get our second theorem.

\begin{restatable}{theorem}{TheoremAbsolute}
	\label{Theorem:Absolute}
	Let functions $f,g : \zo^n \rightarrow \mo$.
	There exists a randomized procedure using $\IQP$ circuits and making two quantum queries to $O_{f,g}$ that accepts with probability
    \begin{align*}
		P_{acc} & = \frac{1}{2} + \frac{1}{4} \Phi^2(f,g) & \text{if } n \ \odd \\
		P_{acc} & = \frac{1}{2} + \frac{1}{8} \Phi^2(f,g) & \text{if } n \ \even 
	\end{align*}
\end{restatable}

\subsubsection{Application to the Raz--Tal oracle separation}

The Raz--Tal separation between $\BQP$ and $\PH$~\cite{RT22} is based on an instantiation of $2$-\textsc{Forrelation}.
Let $\mathcal{F}$ be the set of functions from $\zo^n$ to $\{-1,1\}$, let $\mathcal{U}$ be the uniform distribution on $\mathcal{F}^2$, and let $\mathcal{D}$ be the specific distribution on $\mathcal{F}^2$ constructed by Raz and Tal.
The distinguishing task is: given oracle access to $(f,g)$ drawn from either $\mathcal{D}$ or $\mathcal{U}$, determine which is the case.

For the Raz--Tal distribution, we show that
$$\mathbb{E}_{(f,g)\leftarrow\mathcal{D}}[f(x)g(y)]= (-1)^{x\cdot y} c_N$$
for some $c_N>0$. It follows that
$$\mathbb{E}_{(f,g)\leftarrow\mathcal{D}}[\Phi_{\mathrm{odd}}(f,g)]=\mathbb{E}_{(f,g)\leftarrow\mathcal{D}}[\Phi_{\mathrm{even}}(f,g)]=\frac{1}{2}\mathbb{E}_{(f,g)\leftarrow\mathcal{D}}[\Phi(f,g)]=\Omega\left(\frac{1}{n}\right).$$
On the other hand,
$$\mathbb{E}_{(f,g)\leftarrow\mathcal{U}}[\Phi_{\mathrm{odd}}(f,g)]=\mathbb{E}_{(f,g)\leftarrow\mathcal{U}}[\Phi_{\mathrm{even}}(f,g)]=0.$$

Therefore, unlike in the general case, we do not need to choose randomly between the two IQP circuits.
For odd $n$, Proposition~\ref{Proposition:PhiOdd} gives an $\IQP$ circuit performing a single query to $O_{f,g}$ that accepts with probability
$$P_{\text{acc}}(f,g)=\frac{1}{2}+\frac{1}{\sqrt{2}}\Phi_{\mathrm{odd}}(f,g),$$
and hence distinguishes $\mathcal{D}$ from $\mathcal{U}$ with advantage $\Omega(1/n)$.

In order to conclude, we consider the problem of distinguishing between $m$ pairs of functions taken from $\mathcal{D}^{\otimes m}$ or $\mathcal{U}^{\otimes m}$ with $m=\Theta(n^2)$.
We run $m$ independent copies of the IQP circuit in parallel.
Their tensor product is again a single IQP circuit, and we choose an accepting set that accepts according to the number of accepting copies.
Since Raz and Tal showed that no $\PH$ algorithm can solve this problem with a non-negligible advantage using polynomially many queries, we obtain the following.

\begin{theorem}
    There exists an $\IQP$ circuit that solves the $m$-fold Raz--Tal distinguishing problem with $m=\Theta(n^2)$ and constant advantage, implying that there exists an oracle $O$ such that
    $$\IQP^O \not\subseteq \PH^O.$$
    Moreover, all oracle queries are non-adaptive.
\end{theorem}

In addition, combining our $\IQP$ upper bound with the $\DQC_1$ query lower bound of Girish~\cite{Gir26} yields an oracle $O$ such that $\IQP^O \not\subseteq \DQC_1^O$.

\subsubsection{Fourier growth bounds for $\IQP$ circuits}

We think of the oracle $O_{f,g}$ as querying the string $z \in \{-1,1\}^{\zo^{n+1}}$ encoding the truth tables of $f$ and $g$: $z_{0x} = f(x)$ and $z_{1x} = g(x)$ for $x \in \zo^n$. A single query to $O_{f,g}$ corresponds to one access to this string, and the acceptance probability $p(z)$ of a $d$-query algorithm is a multilinear polynomial in $z$. The level-$\ell$ Fourier growth $L_{1,\ell}(p) = \sum_{|S|=\ell} |\widehat{p}(S)|$ quantifies the size of the level-$\ell$ coefficients of this polynomial.

We prove an accepting-space bound for the acceptance probability $p$ of a $d$-query $\IQP$ circuit with accepting set $F$ which, to the best of our knowledge, is new.
Namely, for every restriction $\rho$,
$$\left\|\widehat{p_{|\rho}}\right\|_1 \leq \sqrt{|F|}.$$
In particular,
$$L_{1,\ell}(p_{|\rho}) \leq \sqrt{|F|}$$
for every level $\ell$. 

Since solving $2$-\textsc{Forrelation} with constant advantage requires $L_{1,2}(p_{|\rho})=\Omega(\sqrt N)$ for some restriction $\rho$~\cite{BS21,Gir26}, any constant-advantage $\IQP$ circuit for $2$-\textsc{Forrelation} must satisfy $|F|=\Omega(N).$
Our $\IQP$ circuits have an accepting set of size $\Theta(N)$, matching this bound up to constants.
This contrasts with the standard $\BQP$ algorithm for $2$-\textsc{Forrelation}, which accepts on a single outcome.
Thus, while a one-dimensional accepting subspace suffices for $\BQP$, any constant-advantage $\IQP$ circuit requires an accepting subspace of dimension $\Omega(N)=\Omega(2^n)$.
This does not make the post-processing inefficient: in our construction, membership in $F$ depends only on the Hamming weight of the measured output modulo $4$, and can therefore be decided in $O(n)$ classical time.

We also derive the level-wise bound
$$L_{1,\ell}(p_{|\rho})\leq\sum_{\substack{a+b=\ell\\0\leq a,b\leq d}}\sqrt{\min\left\{\binom{d}{a}\binom{N}{b},\binom{d}{b}\binom{N}{a}\right\}},$$
and hence, for every fixed $\ell$,
$$L_{1,\ell}(p_{|\rho}) = O_\ell\!\left(d^{\lceil \ell/2\rceil/2}N^{\lfloor \ell/2\rfloor/2}\right).$$
In particular,
$$L_{1,2}(p|_\rho) = O(\sqrt{dN}), \quad L_{1,3}(p|_\rho) = O(d\sqrt{N}), \quad L_{1,6}(p|_\rho) = O(d^{3/2}N^{3/2}).$$
We obtained this bound independently and later observed that its more precise binomial form is already implicit in the non-adaptive analysis of Girish, Sinha, Tal, and Wu~\cite{GSTW24}.
We give a short direct proof specialized to $\IQP$, which retains the sharper asymmetric estimates appearing before the final geometric-mean relaxation in their analysis.

\subsection{Discussion}\label{Section:Discussion}

The $2$-\textsc{Forrelation} problem is an important problem in quantum complexity: it achieves the optimal quantum-vs-classical query separation~\cite{AA15,SSW21} and its power lies beyond the polynomial hierarchy~\cite{RT22}, yet it is not believed to be $\BQP$-complete.
Our construction places $2$-\textsc{Forrelation} in $\BPP^{\IQP}$: the only classical randomness is used to choose between two one-query $\IQP$ circuits.
This is a significantly more restricted model than $\BQP$.
For the Raz--Tal distribution, the classical randomization is unnecessary: we showed that a fixed $\IQP$ circuit already achieves inverse-polynomial distinguishing advantage, leading to an oracle separation $\IQP^O \not\subseteq \PH^O$.

A natural question is whether the power of $\IQP$ circuits extends to higher orders of the Forrelation problem.
For $k>2$, combining Fourier-growth bounds with Theorems~3.2 and~3.4 of~\cite{BS21} shows that any $\IQP$ circuit achieving constant advantage for $k$-\textsc{Forrelation} requires $\Omega_k\!\left(N^{1-2/k}\right)$ queries.
In particular, $3$-\textsc{Forrelation} cannot be solved with constant advantage by $\IQP$ circuits making $d=o(N^{1/3})$ queries.
This is consistent with the more general separation obtained from the fact that $\IQP$ is weaker than $\hBQP$~\cite{JM24} and that $3$-\textsc{Forrelation} cannot be solved in $\hBQP$~\cite{Gir26}.
Thus $\IQP$ remains far from the full power of $\BQP$.

One interesting application is for demonstrating quantum advantage.
$\IQP$ circuits have been proposed as a possible near-term solution to quantum advantage via sampling problems~\cite{BJS11,BMS16,PLS+24}.
However, verifying the output of a sampling experiment can be essentially as hard as simulating it~\cite{AC17}, and attempts to build verifiable $\IQP$ protocols using hidden secrets have proven fragile~\cite{Kah23,BCJ25}.
Because $2$-\textsc{Forrelation} is a \emph{decision} problem, it bypasses these verification bottlenecks.
By choosing instances where the oracles $O_{f,g}$ admit polynomial-size diagonal implementations, our result could lead to new proposals for quantum advantage experiments that are both implementable quantumly and classically verifiable.
Finding such instances remains an open problem.

	\section{Preliminaries}\label{Section:2}
	\subsection{Fourier analysis on $\zo^n$}
	We will be working on the boolean cube $\zo^n$. For $x \in \zo^n$, its Hamming weight $|x|$ is the number of $1$'s contained in the string $x$. For a set $S \subseteq \zo^n$, we denote by $\one_S$ its indicator function. For a function $f : \zo^n \rightarrow \mathbb{C}$, we define its Fourier transform
	$$ \hf(x) = \frac{1}{\sqrt{2^n}} \sum_{y \in \zo^n} (-1)^{x \cdot y} f(y),$$
	where $x \cdot y = \sum_{i=1}^n x_i y_i \mod2$ is the canonical inner product on $\zo^n$. With this normalization, we have the following properties:
	\begin{itemize}
		\item $H^{\otimes n} \left(\sum_{x \in \zo^n} f(x) \ket{x}\right) = \sum_{x \in \zo^n} \hf(x) \ket{x}$ where $H$ is the Hadamard gate. 
		\item (Parseval's identity) $\sum_{x \in \zo^n} |f(x)|^2 = \sum_{x \in \zo^n} |\hf(x)|^2$.
		\item $\widehat{\widehat{f} \ } = f$.
	\end{itemize}
	We also have the useful relations.
	\begin{claim}
		For each $x \in \zo^n$, 
		\begin{align*}
			\sum_{y \in \zo^n} (-1)^{x \cdot y} & = 
			\left\{ 
			\begin{array}{cl}
				2^n & \text{if } x = 0^n \\
				0 & \text{otherwise}
			\end{array}
			\right. \\
			\sum_{\substack{y \in \zo^n \\ |y| \even}} (-1)^{x \cdot y} & = 
			\left\{ 
			\begin{array}{cl}
				2^{n-1} & \text{if } x = 0^n \vee x = 1^n\\
				0 & \text{otherwise}
			\end{array}
			\right. \\
			\sum_{\substack{y \in \zo^n \\ |y| \odd}} (-1)^{x \cdot y} & = 
			\left\{ 
			\begin{array}{cl}
				2^{n-1} & \text{if } x = 0^n \\
				-2^{n-1} & \text{if } x = 1^n \\
				0 & \text{otherwise}
			\end{array}
			\right. 
		\end{align*}
	\end{claim}
	
	\subsection{Fourier expansion on oracle bits}
	We will also view acceptance probabilities as functions of the oracle bits.
	Any function $f:\{-1,1\}^{\zo^n}\to\mathbb{R}$ admits a unique multilinear expansion
	$$f(x)=\sum_{S\subseteq\zo^n}\widehat f(S)\chi_S(x),\quad\chi_S(x)=\prod_{i\in S}x_i.$$
	We measure the distribution of the Fourier mass of $f$ across the levels $|S|$ through the following quantities.
	\begin{definition}($\alpha$-signed level-$\ell$ Fourier growth)
	For $\ell\ge 0$, and signs $\alpha_S\in\{\pm 1\}$, for $S\subseteq\{0,1\}^n$ with $|S|=\ell$, the $\alpha$-signed level-$\ell$ Fourier growth of $f$ is
	$$L_{1,\ell}^\alpha(f) = \sum_{\substack{S\subseteq\zo^n\\ |S|=\ell}} \alpha_S \widehat{f}(S).$$
	\end{definition}
	\begin{definition}[Level-$\ell$ Fourier growth]
		For $\ell\ge 0$, the level-$\ell$ Fourier growth of $f$ is
		$$L_{1,\ell}(f) = \sum_{\substack{S\subseteq\zo^n\\ |S|=\ell}} |\widehat{f}(S)| = \max_{\alpha} L_{1,\ell}^\alpha(f).$$
	\end{definition}
	We always use the natural lexicographic order on binary strings, equivalently identifying them with the corresponding integers written in binary.
	
	\subsection{Matrix norms}
	For a complex matrix $A=(a_{i,j})\in\mathbb{C}^{N\times N}$, we denote by
	$$\norm{A}_F=\left(\sum_{i,j}|a_{i,j}|^2\right)^{1/2}$$
	its Frobenius norm, and by
	$$\norm{A}_{op}=\sup_{\norm{x}_2=1}\norm{Ax}_2$$
	its operator norm.
	\\
	We will use the following standard facts in section~\ref{sec:fourier}:
	\begin{itemize}
		\item For $v\in\mathbb{C}^N$, $\norm{Av}_2^2\leq\norm{A}_{\text{op}}^2\norm{v}_2^2$.
		\item For $A\in\mathcal{C}^{N\times N}$, $\norm{A}_F\leq\sqrt{N}\,\norm{A}_{op}$.
		\item For $A,B\in\mathbb{C}^{N\times N}$, $|Tr(AB)| \leq \norm{A}_F\norm{B}_F$.
		\item For $P$ an orthogonal projector of rank $r$, $\norm{P}_F=\sqrt{r},\ \norm{P}_{op}=1$.
	\end{itemize}
	
	\subsection{$\IQP$ related definitions}
	
	\begin{definition}[IQP circuit]
		An $\IQP$ circuit on $m$ qubits is a quantum circuit of the form $C = H^{\otimes m} D H^{\otimes m}$, whose diagonal block $D$ is a polynomial-size circuit composed of gates diagonal in the computational basis.
		In the relativized setting, an $\IQP^O$ circuit may additionally contain diagonal phase-query gates to the oracle $O$ in the block $D$.
	\end{definition}

	\begin{definition}[$\IQP$]
		A language $L$ is in $\IQP$ if there exists a polynomial-time classical
		algorithm that, given an input $x\in\zo^*$, constructs a polynomial-size $\IQP$
		circuit $C_x$.
		For every input length $n$, we fix a set $F_n$ of accepting measurement
		outcomes, such that membership in $F_n$ can be decided in polynomial time.
		$$\Pr[C_x \text{ outputs } z  \in F_n] \ge \frac{2}{3}\quad\text{if } x\in L,$$
		and
		$$\Pr[C_x \text{ outputs } z \in F_n] \le \frac{1}{3} \quad\text{if } x\notin L,$$
		where $n=|x|$.
		
		The relativized class $\IQP^O$ is defined analogously using $\IQP^O$ circuits.
	\end{definition}

	This acceptance convention is consistent with Girish~\cite{Gir26}, who defines $\BQP$, $\DQC_1$ and $\DQC_k$ algorithms using an arbitrary set of accepting final basis states, and $\hBQP$ using an arbitrary set of accepting initial/final state pairs.
	
	\COMMENT{
	\begin{definition}[IQP oracle]
		The $\IQP$ oracle $\mathcal{O}_{\IQP}$ takes as input a description $w$ of an $\IQP$ circuit on $m$ qubits and returns a single sample $y \in \{0,1\}^m$ drawn from the output distribution of that circuit:
		$$ \Pr[y \mid w] = \left|\bra{y} H^{\otimes m} D_w H^{\otimes m} \ket{0^m}\right|^2.$$
	\end{definition}
	
	\begin{definition}[$\BPP^{\IQP}$]
		A language $L$ is in $\BPP^{\IQP}$ if there exists a randomized polynomial-time classical algorithm $A$ with access to $\mathcal{O}_{\IQP}$ such that for every $x \in \{0,1\}^*$, $\Pr[A^{\mathcal{O}_{\IQP}}(x) = 1] \ge \frac{2}{3}$ if $x \in L$ and $\Pr[A^{\mathcal{O}_{\IQP}}(x) = 1] \le \frac{1}{3}$ otherwise. The probability is over the internal randomness of $A$ and the responses of $\mathcal{O}_{\IQP}$.
	\end{definition}
	
	A $\BPP^\IQP$ circuit is what we will call an $\IQP$ computation. In our case, each $\IQP$ circuit will call the oracle $O_{f,g}$ exactly once so the number of calls to the $\IQP$ oracle will be equal to the number of calls to the $O_{f,g}$ oracle. 
	}

	\section{A single query IQP circuits for $2$-\textsc{Forrelation}}
\label{Section:3}

In this section we prove our main theorem, which we formally restate.

\TheoremMain*

\subsection{Accepting probability of $\IQP$ circuits}
\label{Section:IQPCircuits}

As discussed in the introduction, $\IQP$ circuits are of the following form

$$
\begin{quantikz}
	\lstick{$\ket{0^{n+1}}$} & 
	\gate{H^{\otimes (n+1)}} & \gate{O_{f,g}} & \gate{D} & 
	\gate{H^{\otimes (n+1)}} & \gate{\Pi_F}
\end{quantikz}
$$

\noindent where $D = \sum_{x \in \zo^{n+1}} \rho(x)\kb{x}$ with $|\rho| = 1$ is a diagonal operator in the computational basis.
The circuit is measured in the computational basis and accepts when the outcome belongs to an accepting set $F\subseteq\{0,1\}^{n+1}$.
We denote the corresponding accepting projector by $\Pi_F$.
A circuit of this form has two degrees of freedom: the choice of $D$ and the choice of $F$. We restate the key proposition from the introduction, which identifies sufficient conditions on $D$ and $F$ for the accepting probability to take a useful form.

\PropositionMain*

\begin{proof}
	Start with functions $\rho_0,\rho_1$ and $\sigma$ satisfying the stated conditions. Let 
	$$D = \sum_{x \in \zo^n} {\rho_0(x)} \kb{0x} + {\rho_1(x)} \kb{1x}.$$
	By assumption, $D$ has an efficient implementation using diagonal gates. Let $T = \{x : \hsigma(x) = 1\}$ meaning $\hsigma = (-1)^{\one_{\overline{T}}}$ and $F = \{0\} \times T \cup \{1\} \times \overline{T}$. Since $\hsigma$ is efficiently computable, membership in $F$ can be decided efficiently. We now compute the accepting probability of this circuit.
	
	After the first Hadamards, we get $\frac{1}{\sqrt{2^{n+1}}} \sum_{b \in \zo,\, x \in \zo^n} \ket{b}\ket{x}$. After applying $O_{f,g}$ and then $D$, we obtain
	$$ \frac{1}{\sqrt{2^{n+1}}} \sum_{x \in \zo^n} \rho_0(x) f(x) \ket{0}\ket{x} + \rho_1(x) g(x) \ket{1}\ket{x}.$$
	We set $a(x) = \rho_0(x) f(x)$ and $b(x) = \rho_1(x) g(x)$, so the state is
	$$ \frac{1}{\sqrt{2^{n+1}}} \sum_{x \in \zo^n}  a(x) \ket{0}\ket{x} + b(x) \ket{1}\ket{x}.$$
	Applying the final $H^{\otimes(n+1)}$, we get
	$$ \frac{1}{\sqrt{2^{n+2}}} \sum_{x \in \zo^n}  (\ha(x) + \hb(x)) \ket{0}\ket{x} + (\ha(x) - \hb(x)) \ket{1}\ket{x},$$
	Notice that $\ha$ and $\hb$ are real valued.  The accepting probability is then
	\begin{align*}
		P_{\text{acc}} & = \frac{1}{2^{n+2}} \left(\sum_{x \in T} (\ha(x) + \hb(x))^2 + \sum_{x \in \overline{T}} (\ha(x) - \hb(x))^2\right) \\
		& = \frac{1}{2^{n+2}} \sum_{x \in \zo^n} \ha^2(x) + \hb^2(x) + 2 \sum_{x \in T} \ha(x)\hb(x) - 2 \sum_{x \in \overline{T}} \ha(x)\hb(x).
	\end{align*}
	Since $a,b : \zo^n \to \mo$, Parseval's identity gives $\sum_x \ha^2(x) = \sum_x \hb^2(x) = 2^n$. Recalling that $\hsigma = (-1)^{\one_{\overline{T}}}$, we get
	$$P_{\text{acc}} = \frac{1}{2} + \frac{1}{2^{n+1}}\sum_{x \in \zo^n} \ha(x)\hb(x)\hsigma(x).$$
	To conclude, we expand via the convolution formula:
	\begin{align*}
		\sum_{x \in \zo^n} \ha(x) \hb(x) \hsigma(x)
		& = \frac{1}{2^{3n/2}}\sum_{r,s,t \in \zo^n} \sum_{x \in \zo^n}(-1)^{x \cdot (r + s + t)} a(r) b(s) \sigma(t) \\
		& = \frac{1}{2^{n/2}} \sum_{\substack{r,s,t \\ r+s+t = 0^n}} a(r)b(s)\sigma(t)
		= \frac{1}{2^{n/2}} \sum_{x,y \in \zo^n} a(x)b(y)\sigma(x + y).
	\end{align*}
	Substituting $a(x) = \rho_0(x)f(x)$ and $b(x) = \rho_1(x)g(x)$, we obtain
	$$P_{\text{acc}} = \frac{1}{2} + \sqrt{\frac{1}{2^{3n+2}}}\sum_{x,y \in \zo^n} f(x)g(y){\rho_0(x) \rho_1(y)}\sigma(x+y). \qedhere$$
\end{proof}

In order to get an advantage for the $2$-\textsc{Forrelation} problem, we would like to find functions $\rho_0,\rho_1,\sigma$ such that $\rho_0(x)\rho_1(y)\sigma(x+y) = (-1)^{x \cdot y}$ for all $x,y$. We show this is not possible.

\subsection{Impossibility result and parity splitting}

\begin{proposition}
	There do not exist functions $\rho_0,\rho_1 : \zo^n \rightarrow \mo$ and $\sigma : \zo^n \rightarrow \mathbb{R}$ such that 
	$$ \forall x,y \in \zo^n, \quad {\rho_0(x) \rho_1(y)}\sigma(x+y) = (-1)^{x \cdot y}.$$
\end{proposition}
\begin{proof}
	Assume for contradiction that such a triple $(\rho_0,\rho_1,\sigma)$ exists. Since $\sigma(x+y) = (-1)^{x\cdot y}\rho_0(x)\rho_1(y) \in \{-1,1\}$, there exists $\gamma : \zo^n \rightarrow \zo$ with $\sigma = (-1)^\gamma$. Similarly, we write $\rho_b = (-1)^{\alpha_b}$ with $\alpha_b : \zo^n \rightarrow \zo$ for $b \in \{0,1\}$. The relation becomes
	$$x\cdot y = \alpha_0(x) + \alpha_1(y) + \gamma(x+y) \mod2.$$
	Setting $z = x+y$ gives
	\begin{equation}\label{eq:ik}
		x\cdot(x+z) = \alpha_0(x) + \alpha_1(x+z) + \gamma(z) \mod2,
	\end{equation}
	and setting $z = 0^n$ gives
	\begin{equation}\label{eq:ii}
		|x| = \alpha_0(x) + \alpha_1(x) + \gamma(0^n) \mod2.
	\end{equation}
	Adding \eqref{eq:ii} and \eqref{eq:ik} yields
	\begin{equation}\label{eq:rho1}
		x\cdot z = \alpha_1(x) + \alpha_1(x+z) + \gamma(z) + \gamma(0^n) \mod2.
	\end{equation}
	Replacing $x$ by $x+z$ in \eqref{eq:rho1} gives
	\begin{equation}\label{eq:rho1k}
		x\cdot z + |z| = \alpha_1(x) + \alpha_1(x+z) + \gamma(z) + \gamma(0^n) \mod2.
	\end{equation}
	Adding \eqref{eq:rho1} and \eqref{eq:rho1k} yields $|z| \equiv 0 \mod2$ for all $z$, a contradiction.
\end{proof}

The obstruction comes from Hamming weight parity. However, if we restrict to pairs $(x,y)$ where $|x+y|$ has a fixed parity, we can circumvent the above impossibility result. We define
\begin{align*}
	\Phi_{\odd}(f,g) & = \frac{1}{2^{3n/2}}\sum_{\substack{x,y \in \zo^n \\ |x + y| \ \odd}} f(x)g(y)(-1)^{x \cdot y}, \\
	\Phi_{\even}(f,g) & = \frac{1}{2^{3n/2}}\sum_{\substack{x,y \in \zo^n \\ |x + y| \ \even}} f(x)g(y)(-1)^{x \cdot y}.
\end{align*}
Note that $\Phi(f,g) = \Phi_{\odd}(f,g) + \Phi_{\even}(f,g)$. Our strategy is to find functions that give a bias proportional to $\Phi_{\odd}(f,g)$ and, separately, to $\Phi_{\even}(f,g)$, then combine the two circuits.

\subsection{The quadratic function $Q$}

\begin{definition}
	For $x = (x_1,\dots,x_n) \in \zo^n$, define $Q(x) = \sum_{1 \le i < j \le n} x_ix_j$.
\end{definition}

The function $Q$ satisfies $Q(x) \equiv \binom{|x|}{2} \mod2$, so it depends only on the Hamming weight. $x \mapsto (-1)^{Q(x)}$ is the simplest example of a bent function over $\zo^n$ (for $n$ even). The key identity we exploit is the following.

\begin{proposition}\label{prop:Q-identity}
	For all $x, y \in \zo^n$,
	\[
	Q(x) + Q(y) + Q(x+y) = x \cdot y + |x||y| \mod2.
	\]
\end{proposition}

\begin{proof}
	Expanding directly,
	\begin{align*}
		Q(x) + Q(y) + Q(x + y) & = \sum_{1 \le i < j \le n} x_ix_j + y_iy_j + (x_i + y_i)(x_j + y_j) \mod2 \\
		& = \sum_{1 \le i < j \le n} (x_i y_j + y_i x_j) \mod2 \\
		& 
		= \Bigl(\sum_{i} x_i\Bigr)\Bigl(\sum_{j} y_j\Bigr) - \sum_{i} x_iy_i \mod2 \\
		& = |x||y| + x \cdot y \mod2. \qedhere
	\end{align*}
\end{proof}

In particular, when $|x+y|$ is odd, at least one of $|x|$, $|y|$ must be even, so $|x||y| \equiv 0 \mod2$ and Proposition~\ref{prop:Q-identity} gives $Q(x)+Q(y)+Q(x+y) = x \cdot y$. This is the key identity that makes $Q$ useful for our construction.

\begin{remark}
	A $cZ$ gate on qubits $i,j$ applies the phase $(-1)^{x_i x_j}$ to the computational basis state $\ket{x}$.
	Therefore the phase $(-1)^{Q(x)}$ is implemented using $\binom{n}{2}=O(n^2)$ commuting diagonal two-qubit gates using that
	$$
    \prod_{1\le i<j\le n} cZ_{i,j}\ket{x}
    =
    \prod_{1\le i<j\le n}(-1)^{x_i x_j}\ket{x}
    =
	(-1)^{\sum_{1\le i<j\le n} x_ix_j}\ket{x}
	=
    (-1)^{Q(x)}\ket{x}.
	$$
\end{remark}

\subsection{Main theorem for $n$ odd}

\subsubsection{$n$ odd, $\Phi_{\odd}(f,g)$}

\begin{proposition}\label{Proposition:PhiOdd}
	Let $n$ be an odd integer and let $f,g : \zo^n \rightarrow \mo$. There exists an $\IQP$ circuit performing a single query to $O_{f,g}$ that accepts with probability $\frac{1}{2} + \frac{1}{\sqrt{2}}\Phi_{\odd}(f,g)$.
\end{proposition}
\begin{proof}
	Our goal is to use Proposition~\ref{Proposition:2}. We choose the following functions
	\begin{itemize}
		\item $\rho_0 = \rho_1 = (-1)^Q$.
		\item $\sigma(x) = \left\{
		\begin{array}{cl}
			\sqrt{2} (-1)^{Q(x)} & \text{if } |x| \ \odd \\
			0 & \text{otherwise}
		\end{array}\right.$
	\end{itemize}
	$\rho_0,\rho_1$ are easily computable. We now show that $\sigma$ satisfies the bentness condition and that $\hsigma$ is efficiently computable. Both these conditions are captured by the following lemma
	\begin{lemma}\label{Lemma:OddBent}
		$\hsigma(x) = \sqrt{2}\sin\!\left(\frac{\pi}{4}\left(n - 2|x|\right)\right)$. In particular, when $n$ is odd, $|\hsigma(x)| = 1$ for all $x \in \zo^n$.
	\end{lemma}
	\begin{proof}
		Since $Q(x) \equiv \binom{|x|}{2} \mod2$, $\sigma$ depends only on the Hamming weight. We get $\sigma(x) = \sqrt{2}$ if $|x| \equiv 1 \pmod{4}$ and $\sigma(x) = -\sqrt{2}$ if $|x| \equiv 3 \pmod{4}$, and $\sigma(x) = 0$ otherwise. This can be written in closed form as $\sigma(x) = \sqrt{2}\sin(\frac{\pi}{2}|x|)$, which we rewrite as $\sigma(x) = \sqrt{2}\,\Im(e^{i \frac{\pi}{2} |x|})$, where $\Im(\cdot)$ denotes the imaginary part. We then compute 
		\begin{align*}
			\hsigma(x) & = \sqrt{2}\,\Im \left(\frac{1}{\sqrt{2^n}}\sum_{y \in \zo^n} (-1)^{x \cdot y}e^{i\frac{\pi}{2}|y|}\right) \\
			& = \sqrt{2}\,\Im \left(\frac{1}{\sqrt{2^n}}\sum_{y \in \zo^n} \prod_{j=1}^n (-1)^{x_j y_j} e^{i\frac{\pi}{2} y_j}\right) \\
			& = \sqrt{2}\,\Im \left(\frac{1}{\sqrt{2^n}} \prod_{j = 1}^n \left(1 + (-1)^{x_j} e^{i \frac{\pi}{2}}\right)\right) \\
			& = \sqrt{2}\,\Im\left(\frac{1}{\sqrt{2^n}} \left(1 - i\right)^{|x|} \left(1 + i\right)^{n - |x|}\right) 
			\\ & =  \sqrt{2}\,\Im\left(\left(e^{i\frac{\pi}{4}}\right)^{n - 2|x|}\right) \\
			& = \sqrt{2}\,\sin\left(\frac{\pi}{4}\left(n - 2|x|\right)\right).
		\end{align*} 
	\end{proof}
	
	Since $|\hsigma(x)| = 1$ for all $x$, the image of $\hsigma$ is $\{-1,1\}$, and $\hsigma$ is efficiently computable. We can apply Proposition~\ref{Proposition:2} to obtain an $\IQP$ circuit with accepting probability
	\begin{align*}
		P_{\text{acc}} & = \frac{1}{2} + \sqrt{\frac{1}{2^{3n+2}}} \sum_{x,y \in \zo^n} f(x)g(y) \rho_0(x)\rho_1(y)\sigma(x + y) \\
		& = \frac{1}{2} + \sqrt{\frac{1}{2^{3n+1}}} \sum_{\substack{x,y \in \zo^n \\ |x+y| \ \odd}} f(x)g(y) (-1)^{Q(x) + Q(y) + Q(x+y)} \\
		& = \frac{1}{2} + \sqrt{\frac{1}{2^{3n+1}}} \sum_{\substack{x,y \in \zo^n \\ |x+y| \ \odd}} f(x)g(y) (-1)^{x \cdot y + |x||y|} \\
		& = \frac{1}{2} + \sqrt{\frac{1}{2^{3n+1}}} \sum_{\substack{x,y \in \zo^n \\ |x+y| \ \odd}} f(x)g(y) (-1)^{x \cdot y} \\
		& = \frac{1}{2} + \frac{1}{\sqrt{2}}\Phi_{\odd}(f,g),
	\end{align*}
	where in the fourth line we used that $|x+y|$ odd implies $|x||y| \equiv 0 \mod2$. \qedhere
\end{proof}

\subsubsection{$n$ odd, $\Phi_{\even}(f,g)$}

\begin{proposition}\label{Proposition:PhiEven}
	Let $n$ be odd and let $f,g : \zo^n \rightarrow \mo$. There exists an $\IQP$ circuit performing a single query to $O_{f,g}$ that accepts with probability $\frac{1}{2} + \frac{1}{\sqrt{2}}\Phi_{\even}(f,g)$.
\end{proposition}
\begin{proof}
	To capture $\Phi_{\even}$, we need $\sigma$ supported on even-weight inputs. When $|x+y|$ is even, $|x| \equiv |y| \mod2$, so $|x||y| \equiv |y|^2 \equiv |y| \mod2$, which means the extra term $|x||y|$ in Proposition~\ref{prop:Q-identity} can be cancelled by modifying $\rho_1$. We choose:
	\begin{itemize}
		\item $\rho_0(x) = (-1)^{Q(x)}$, \quad $\rho_1(x) = (-1)^{Q(x) + |x|}$.
		\item $\sigma(x) = \begin{cases} \sqrt{2}\,(-1)^{Q(x)} & \text{if } |x| \text{ is even,} \\ 0 & \text{otherwise.} \end{cases}$
	\end{itemize}
	\begin{lemma}\label{Lemma:EvenBent}
		$\hsigma(x) = \sqrt{2}\cos\left(\frac{\pi}{4}\left(n - 2|x|\right)\right)$. In particular, when $n$ is odd, $|\hsigma(x)| = 1$ for all $x \in \zo^n$.
	\end{lemma}
	\begin{proof}
		We perform a similar calculation to Lemma~\ref{Lemma:OddBent}.  Using $Q(x) \equiv \binom{|x|}{2} \mod2$, we write $\sigma(x) = \sqrt{2}$ if $|x| \equiv 0 \pmod{4}$, $\sigma(x) = -\sqrt{2}$ if $|x| \equiv 2 \pmod{4}$, and $\sigma(x) = 0$ otherwise. This can be written in closed form as $\sigma(x) = \sqrt{2}\cos(\frac{\pi}{2}|x|)$, which we rewrite as $\sigma(x) = \sqrt{2}\,\Re(e^{i \frac{\pi}{2} |x|})$, where $\Re(\cdot)$ denotes the real part. The same computation as in Lemma~\ref{Lemma:OddBent} gives  
		\begin{align*}
			\hsigma(x) & = \sqrt{2}\,\Re\left(\frac{1}{\sqrt{2^n}}\sum_{y \in \zo^n} (-1)^{x \cdot y}e^{i\frac{\pi}{2}|y|}\right) \\
			& = \sqrt{2}\,\Re \left(\frac{1}{\sqrt{2^n}}\sum_{y \in \zo^n} \prod_{j=1}^n (-1)^{x_j y_j} e^{i\frac{\pi}{2} y_j}\right) \\
			& = \sqrt{2}\,\Re\left(\frac{1}{\sqrt{2^n}} \prod_{j = 1}^n \left(1 + (-1)^{x_j} e^{i \frac{\pi}{2}}\right)\right) \\
			& = \sqrt{2}\,\Re\left(\frac{1}{\sqrt{2^n}} \left(1 - i\right)^{|x|} \left(1 + i\right)^{n - |x|}\right) \\
			& =  \sqrt{2}\,\Re\left(\left(e^{i\frac{\pi}{4}}\right)^{n - 2|x|}\right) \\
			& = \sqrt{2}\cos\left(\frac{\pi}{4}\left(n - 2|x|\right)\right).
		\end{align*} 
	\end{proof}
	Applying Proposition~\ref{Proposition:2}, we have an $\IQP$ circuit that accepts with probability
	\begin{align*}
		P_{\text{acc}} & = \frac{1}{2} + \sqrt{\frac{1}{2^{3n+2}}} \sum_{x,y \in \zo^n} f(x)g(y) (-1)^{Q(x) + Q(y) + |y|}\sigma(x + y) \\
		& = \frac{1}{2} + \sqrt{\frac{1}{2^{3n+1}}} \sum_{\substack{x,y \in \zo^n \\ |x+y| \ \even}} f(x)g(y) (-1)^{Q(x) + Q(y) + Q(x+y) + |y|} \\
		& = \frac{1}{2} + \sqrt{\frac{1}{2^{3n+1}}} \sum_{\substack{x,y \in \zo^n \\ |x+y| \ \even}} f(x)g(y) (-1)^{x \cdot y + |x||y| + |y|} \\
		& = \frac{1}{2} + \sqrt{\frac{1}{2^{3n+1}}} \sum_{\substack{x,y \in \zo^n \\ |x+y| \ \even}} f(x)g(y) (-1)^{x \cdot y}
		= \frac{1}{2} + \frac{1}{\sqrt{2}}\Phi_{\even}(f,g),
	\end{align*}
	where in the fourth line we used that $|x+y|$ even implies $|x| + |y| \equiv 0 \mod2$ and thus $|x||y| + |y| \equiv |y|(|y|+1) \mod 2 \equiv 0 \mod2$. \qedhere
\end{proof}

\subsection{Proving the main theorem}

\begin{proposition}\label{Proposition:MainOdd}
	Let $n$ be odd. Let $f,g : \zo^{n} \rightarrow \mo$.
    There exists a randomized procedure that chooses uniformly between two $\IQP$ circuits, each performing a single query to $O_{f,g}$, and accepts with probability
	$$ P_{\text{acc}} = \frac{1}{2} + \frac{1}{2\sqrt{2}}\Phi(f,g).$$
\end{proposition}
\begin{proof}
	Draw a uniform random bit $r \in \zo$. If $r = 0$, run the circuit of Proposition~\ref{Proposition:PhiOdd}; if $r = 1$, run the circuit of Proposition~\ref{Proposition:PhiEven}. The overall accepting probability is
	$$ P_{\text{acc}} = \frac{1}{2}\!\left(\frac{1}{2} + \frac{1}{\sqrt{2}}\Phi_{\odd}(f,g)\right) + \frac{1}{2}\!\left(\frac{1}{2} + \frac{1}{\sqrt{2}}\Phi_{\even}(f,g)\right) = \frac{1}{2} + \frac{1}{2\sqrt{2}}\Phi(f,g). \qedhere$$
\end{proof}

\paragraph{The $n$ even case.}
When $n$ is even, the function $\sigma$ constructed above no longer satisfies $|\hsigma| = 1$ everywhere, so the bentness condition in Proposition~\ref{Proposition:2} fails. We handle this by reducing to the odd case via a one-bit padding argument.

\begin{definition}
	Given $f,g : \zo^n \rightarrow \mo$, define $\bar{f}, \bar{g} : \zo^{n+1} \rightarrow \mo$ by
	$$ \bar{f}(x_1, \dots, x_{n+1}) = f(x_1, \dots, x_n), \qquad \bar{g}(x_1, \dots, x_{n+1}) = g(x_1, \dots, x_n).$$
\end{definition}

\begin{proposition}\label{Proposition:MainEven}
	Let $n$ be even. Let $f,g : \zo^{n} \rightarrow \mo$.
    There exists a randomized procedure that chooses uniformly between two $\IQP$ circuits, each performing a single query to $O_{f,g}$, and accepts with probability
	$$ P_{\text{acc}} = \frac{1}{2} + \frac{1}{4}\Phi(f,g).$$
\end{proposition}

\begin{proof}
	Since $n+1$ is odd, we apply Proposition~\ref{Proposition:MainOdd} to $\bar{f}$ and $\bar{g}$, giving a randomized procedure that accepts with probability $P_{\text{acc}} = \frac{1}{2} + \frac{1}{2\sqrt{2}}\Phi(\bar{f},\bar{g})$ and makes a single call to $O_{\bar{f},\bar{g}}$, which can be simulated with a single call to $O_{f,g}$. It remains to relate $\Phi(\bar{f},\bar{g})$ to $\Phi(f,g)$. Writing $x = (x', a)$ and $y = (y', b)$ with $x', y' \in \zo^n$ and $a, b \in \zo$,
	\begin{align*}
		\Phi(\bar{f}, \bar{g}) & = \frac{1}{2^{3(n+1)/2}} \sum_{x', y' \in \zo^n} (-1)^{x' \cdot y'} f(x')\, g(y') \sum_{a, b \in \zo} (-1)^{ab} \\
		& = \frac{2}{2^{3(n+1)/2}} \sum_{x', y' \in \zo^n} (-1)^{x' \cdot y'} f(x')\, g(y') \\
		& = \frac{2}{2\sqrt{2} \cdot 2^{3n/2}} \sum_{x', y' \in \zo^n} (-1)^{x' \cdot y'} f(x')\, g(y')
		\\ & = \frac{1}{\sqrt{2}} \Phi(f,g),
	\end{align*}
	where we used $\sum_{a,b \in \zo} (-1)^{ab} = 2$. Substituting this in $P_{\text{acc}}$, we obtain $P_{\text{acc}} = \frac{1}{2} + \frac{1}{4}\Phi(f,g)$.
\end{proof}

\noindent In order to put everything together, notice that  Propositions~\ref{Proposition:MainOdd} and~\ref{Proposition:MainEven} together give Theorem~\ref{Theorem:Main}.

\subsection{Proof of the second theorem}
\label{Section:Absolute}

\TheoremAbsolute*

\begin{proof}
	Run the randomized procedure of Theorem~\ref{Theorem:Main} twice independently, and accept if and only if both runs give the same outcome (\textit{i.e.} when they both accept or both reject). If $n$ is odd,
	$$ P_{acc} = \left(\frac{1}{2} + \frac{1}{2\sqrt{2}}\Phi(f,g)\right)^2 + \left(\frac{1}{2} - \frac{1}{2\sqrt{2}}\Phi(f,g)\right)^2 = \frac{1}{2} + \frac{1}{4}\Phi^2(f,g).$$
	If $n$ is even,
	$$ P_{acc} = \left(\frac{1}{2} + \frac{1}{4}\Phi(f,g)\right)^2 + \left(\frac{1}{2} - \frac{1}{4}\Phi(f,g)\right)^2 = \frac{1}{2} + \frac{1}{8}\Phi^2(f,g). \qedhere$$
\end{proof}

\begin{remark}
    Our constructions also give quantum algorithms for the $\hBQP$ model, since the coin can be compiled there.
\end{remark}
	
	\section{Oracle Separations}
\label{sec:oracle-separation}

In this section we prove that there exists an oracle $O$ such that $\IQP^O \not\subseteq \PH^O$, strengthening the Raz--Tal separation~\cite{RT22}.
The key observation is that, for the specific distribution of Raz and Tal, the random choice between the odd and even circuits in our general $2$-\textsc{Forrelation} algorithm is unnecessary.

We recall the setup of Raz and Tal~\cite{RT22}.
Let $\mathcal{F}$ be the set of functions from $\zo^n$ to $\{-1,1\}$.
Consider two distributions on $\mathcal{F}^2$: the uniform distribution $\mathcal{U}$, and a structured distribution $\mathcal{D}$ constructed in~\cite{RT22} so that pairs $(f,g) \sim \mathcal{D}$ have large forrelation.

\begin{definition}[Raz--Tal distinguishing problem]
    Given oracle access to a pair $(f,g) \in \mathcal{F}^2$ promised to be drawn from either $\mathcal{D}$ or $\mathcal{U}$, determine which is the case with success probability at least $\frac{2}{3}$.
\end{definition}

\begin{definition}[$m$-fold Raz--Tal distinguishing problem]
    Given oracle access to a pair
    $$((f_1,g_1),\dots,(f_m,g_m)) \in \mathcal{F}^{2m}$$
    promised to be drawn from either $\mathcal{D}^{\otimes m}$ or $\mathcal{U}^{\otimes m}$, determine which is the case with success probability at least $\frac{2}{3}$.
\end{definition}

Let $\varepsilon=\frac{1}{24\ln N} = \frac{1}{24\ln 2}\frac{1}{n}$ be the parameter used in the Raz--Tal distribution.
Their distribution satisfies the following properties.

\begin{lemma}\label{lem:dist1}
    There exists $c_N>0$ such that, for every $x,y\in\zo^n$,
    $$\mathbb{E}_{(f,g) \leftarrow \mathcal{D}}(f(x)g(y)) = (-1)^{x\cdot y}c_N.$$
    Moreover, $c_N$ is such that
    $$c_N\sqrt{N} \geq \frac{\varepsilon}{2}.$$
\end{lemma}

\begin{proof}
Recall that $\mathcal{D}$ is obtained by truncating and rounding a centered Gaussian vector $(X,Y)$ satisfying
$$\mathbb{E}[X_xY_y]=\frac{\varepsilon}{\sqrt N}(-1)^{x\cdot y}.$$
Since the rounding preserves conditional expectations,
$$\mathbb{E}_{(f,g) \leftarrow \mathcal{D}}[f(x)g(y)]=\mathbb{E}[\text{trnc}(X_x)\text{trnc}(Y_y)].$$
The law of $(X_x,Y_y)$ depends on $(x,y)$ only through the sign $(-1)^{x\cdot y}$ of its covariance.
Since $\text{trnc}$ is odd, there exists $c_N$, independent of $x,y$, such that
$$\mathbb{E}_{(f,g) \leftarrow \mathcal{D}}[f(x)g(y)]=(-1)^{x\cdot y}c_N.$$
It follows that
\begin{align*}
\mathbb{E}_{(f,g) \leftarrow \mathcal{D}}[\Forr(f,g)]
&= \frac{1}{N^{3/2}}\sum_{x,y\in\zo^n} (-1)^{x\cdot y}\mathbb{E}_{(f,g) \leftarrow \mathcal{D}}[f(x)g(y)] \\
&= \frac{1}{N^{3/2}}\sum_{x,y\in\zo^n} c_N \\
&= c_N\sqrt N.
\end{align*}
Raz and Tal~\cite{RT22} show that
$$\mathbb{E}_{(f,g) \leftarrow \mathcal{D}}[\Forr(f,g)]\geq \frac{\varepsilon}{2}.$$
Therefore
$$c_N\sqrt N\geq\frac{\varepsilon}{2}$$
which concludes the proof.
\end{proof}

\begin{lemma}\label{lem:dist2}
    For the Raz and Tal distribution $\mathcal{D}$,
    $$\mathbb{E}_{(f,g) \leftarrow \mathcal{D}}\left[\Forr_\odd(f,g)\right]
    =\mathbb{E}_{(f,g) \leftarrow \mathcal{D}}\left[\Forr_\even(f,g)\right]
    =\frac12 \mathbb{E}_{(f,g) \leftarrow \mathcal{D}}\left[\Forr(f,g)\right].$$
    Thus,
    $$\mathbb{E}_{(f,g) \leftarrow \mathcal{D}}\left[\Forr_\odd(f,g)\right]
    =\mathbb{E}_{(f,g) \leftarrow \mathcal{D}}\left[\Forr_\even(f,g)\right]
    \geq \frac{\varepsilon}{4}.$$
    For the uniform distribution $\mathcal{U}$,
    $$\mathbb{E}_{(f,g) \leftarrow \mathcal{U}}\left[\Forr_\odd(f,g)\right]
    =\mathbb{E}_{(f,g) \leftarrow \mathcal{U}}\left[\Forr_\even(f,g)\right]
    =\mathbb{E}_{(f,g) \leftarrow \mathcal{U}}\left[\Forr(f,g)\right]
    = 0.$$
\end{lemma}

\begin{proof}
    Using Lemma~\ref{lem:dist1},
    \begin{align*}
        \mathbb{E}_{(f,g) \leftarrow \mathcal{D}}\left[\Forr_\odd(f,g)\right]
        &= \frac{1}{N^{3/2}} \sum_{\substack{x,y\in\zo^n\\ |x+y|\ \odd}} (-1)^{x\cdot y} \mathbb{E}_{(f,g) \leftarrow \mathcal{D}}(f(x)g(y)) \\
        &= \frac{c_N}{N^{3/2}} \left|\left\{(x,y)\ : \ |x+y|\ \odd \right\}\right|
    \end{align*}
    Remark that $$\left|\left\{(x,y)\ : \ |x+y|\ \odd \right\}\right| = \frac{N^2}{2}.$$
    Thus,
    $$\mathbb{E}_{(f,g) \leftarrow \mathcal{D}}\left[\Forr_\odd(f,g)\right] = \frac{c_N}{N^{3/2}} \frac{N^2}{2} = \frac{c_N \sqrt{N}}{2}.$$
    The same reasoning gives
    $$\mathbb{E}_{(f,g) \leftarrow \mathcal{D}}\left[\Forr_\even(f,g)\right] = \frac{c_N\sqrt{N}}{2}.$$

    Under $\mathcal{U}$, $f(x)$ and $g(y)$ are independent unbiased signs, so $\mathbb{E}_{(f,g) \leftarrow \mathcal{U}}[f(x)g(y)]=0$ for every $x,y$.
    By linearity of expectation, this implies
    $$\mathbb{E}_{(f,g) \leftarrow \mathcal{U}}[\Forr_\odd]=\mathbb{E}_{(f,g) \leftarrow \mathcal{U}}[\Forr_\even]=\mathbb{E}_{(f,g) \leftarrow \mathcal{U}}[\Forr]=0.$$
\end{proof}

Thus, for the Raz--Tal distribution, the forrelation bias is equally shared between the even and odd parts, so either one alone retains the same order of distinguishing bias.
Therefore, unlike in the general case, it suffices to use either the even or the odd IQP circuit alone.

\begin{proposition}\label{prop:1fold}
    There is an $\IQP$ circuit that distinguishes $\mathcal{D}$ from $\mathcal{U}$ with advantage $\Omega(1/n)$.
\end{proposition}

\begin{proof}
    It is sufficient for the oracle separation to consider only the odd values of $n$.
    By Proposition~\ref{Proposition:PhiOdd}, there is an $\IQP$ circuit that accepts with probability 
    $$P_{\text{acc}}(f,g) = \frac{1}{2} + \frac{1}{\sqrt{2}}\Phi_{\odd}(f,g).$$
    Hence, by Lemma~\ref{lem:dist2},
    \begin{align*}
        \left|\mathbb{E}_{(f,g) \leftarrow \mathcal{D}}\left[P_{\text{acc}}(f,g)\right] - \mathbb{E}_{(f,g) \leftarrow \mathcal{U}}\left[P_{\text{acc}}(f,g)\right] \right|
        &= \frac{1}{\sqrt{2}} \left|\mathbb{E}_{(f,g) \leftarrow \mathcal{D}}\left[\Forr_\odd(f,g)\right] - \mathbb{E}_{(f,g) \leftarrow \mathcal{U}}\left[\Forr_\odd(f,g)\right] \right| \\
        &\geq \frac{\varepsilon}{4\sqrt{2}} \\
        &= \Omega(1/n).
    \end{align*}
\end{proof}

\begin{proposition}\label{prop:mfold}
    For $m=\Theta(n^2)$, there is an $\IQP$ circuit using $m$ non-adaptive queries that solves the $m$-fold Raz--Tal distinguishing problem.
\end{proposition}

\begin{proof}		
Let $C$ be the $\IQP$ circuit from Proposition~\ref{prop:1fold}, with accepting set $F$.
Recall that $C$ distinguishes $\mathcal D$ from $\mathcal U$ with advantage $\Omega(1/n)$.
Take the tensor product of $m$ independent copies of $C$ running in parallel.
If $$C=H^{\otimes (n+1)}DH^{\otimes (n+1)},$$
then
$$C^{\otimes m}=H^{\otimes (n+1)m}D^{\otimes m}H^{\otimes (n+1)m}.$$
Since $D^{\otimes m}$ is diagonal, $C^{\otimes m}$ is again an IQP circuit.

We define its accepting set by
$$F' = \left\{(z_1,\dots,z_m) : \frac{1}{m}\sum_{i=1}^{m}\mathbbm{1}[z_i\in F] \geq \frac{1}{2} + \frac{1}{2}\frac{\varepsilon}{4\sqrt{2}} \right\}$$
The accepting projector is
$$\Pi_{F'} = \sum_{z\in F'} \kb{z}.$$
Membership in $F'$ can be decided in polynomial time, since one only needs to test $z_i\in F$ for each of the $m=\Theta(n^2)$ blocks and count the number of accepting blocks.

Since the $m$ instances are independent and each copy has distinguishing gap $\Omega(1/n)$, a standard Chernoff bound amplifies this gap to a constant advantage for $m=\Theta(n^2)$.
The resulting computation is therefore a single IQP circuit solving the $m$-fold Raz--Tal distinguishing problem with constant advantage and using $\Theta(n^2)$ oracle queries in total.
\end{proof}

Raz and Tal~\cite{RT22} proved that this problem is hard for $\PH$.

\begin{proposition}[\cite{RT22}]\label{Theorem:PH}
    Any algorithm in $\PH$ requires $2^{\Omega(n)}$ queries to solve the $m$-Raz--Tal distinguishing problem with a non-negligible advantage for any $m = poly(n)$.
\end{proposition}

We now state our oracle separation. 

\begin{theorem}\label{Theorem:OracleSeparation}
    There exists an oracle $O$ relative to which $\IQP^O \not\subseteq \PH^O$.
\end{theorem}

\begin{proof}
    For $m=\Theta(n^2)$, Proposition~\ref{prop:mfold} gives an $\IQP$ circuit solving the $m$-fold Raz--Tal distinguishing problem with constant advantage.
    By Proposition~\ref{Theorem:PH}, no $\PH$ algorithm can do the same, which gives the oracle separation.
\end{proof}

\begin{remark}
    Since $\IQP\subseteq\BQP$, Theorem~\ref{Theorem:OracleSeparation} strengthens the Raz--Tal separation by showing that the separation already holds for the much more restricted $\IQP$ model.
\end{remark}

\begin{corollary}
    There exists an oracle $O$ such that
    $$\IQP^O \not\subseteq \DQC_1^O.$$
\end{corollary}

\begin{proof}
    Girish's $\DQC_1$ Fourier-growth bound~\cite{Gir26} implies an exponential query lower bound for the Raz--Tal distinguishing problem, whereas our construction solves it using polynomially many queries in a single $\IQP$ circuit.
\end{proof}
	
	\section{Fourier growth bounds for $\IQP$}\label{sec:fourier}
In this section, we study the Fourier spectrum of the acceptance probability of a $d$-query $\IQP$ circuit.
Let $x\in\{-1,1\}^{\zo^n}$, and let the oracle be the phase oracle
$$O_x\ket{a}=x_a\ket{a}, \qquad a\in\zo^n.$$
Consider a $d$-query $\IQP$ circuit on $m$ qubits, including any ancilla qubits.

For a computational-basis state $u\in\zo^m$ and $t\in[d]$, let 
$$\pi_t(u)\in\zo^n$$
denote the oracle address queried by the $t$-th oracle call on $\ket{u}$.
The oracle phase accumulated on $\ket{u}$ is therefore $\prod_{t=1}^d x_{\pi_t(u)}.$
Define
$$A(u) = \triangle_{t=1}^d\{\pi_t(u)\} \subseteq \zo^n,$$
where $\triangle$ denotes symmetric difference.
Since $x_a^2=1$, we have 
$$\prod_{t=1}^d x_{\pi_t(u)} = \chi_{A(u)}(x) \quad \text{where } \chi_A(x):=\prod_{a\in A}x_a.$$
Notice that $|A(u)|\le d$.
This description allows the oracle calls to act on arbitrary, possibly overlapping, subsets of the $\IQP$ register.

$$
\begin{quantikz}
	& \gate[2]{O_x} & & \\
	& & \gate[2]{O_x} & \\
	& \gate[2]{O_x} & & \\
	& & &
\end{quantikz}
$$

Let $D=\sum_{u\in\zo^m}e^{i\varphi_u}\ket{u}\!\bra{u}$ denote the input-independent part of the diagonal layer. Immediately before the final Hadamard layer, the state of the circuit is $\ket{\psi_x} = 2^{-m/2} \sum_{u\in\zo^m} e^{i\varphi_u}\chi_{A(u)}(x)\ket{u}.$
For every $A\subseteq\zo^n$, define 
$$\ket{v_A} = 2^{-m/2} \sum_{\substack{u\in\zo^m\\A(u)=A}} e^{i\varphi_u}\ket{u}.$$
We then obtain the Fourier decomposition 
$$\ket{\psi_x} = \sum_{\substack{A\subseteq\zo^n\\|A|\le d}} \chi_A(x)\ket{v_A}.$$
The vectors $\{\ket{v_A}\}_A$ are pairwise orthogonal, since they are supported on disjoint sets of computational-basis states, and $\sum_A\|v_A\|_2^2=1.$

We will also need this decomposition after fixing some of the oracle bits. A \emph{restriction} is a string $\rho\in\{-1,1,*\}^{\zo^n}.$
Let $\Gamma=\{a\in\zo^n:\rho_a=*\}$ be the set of free oracle coordinates. For a function $f:\{-1,1\}^{\zo^n}\to\mathbb{R}$, we denote by $f|_\rho$ the function obtained by fixing every coordinate $a\notin\Gamma$ to $\rho_a$.

For every $A\subseteq\zo^n$, set 
$$\varepsilon_A(\rho) = \prod_{a\in A\setminus\Gamma}\rho_a.$$
Then 
$$\chi_A|_\rho=\varepsilon_A(\rho)\chi_{A\cap\Gamma}.$$
Grouping together characters that become identical under the restriction gives the following decomposition.

\begin{lemma}[Fourier decomposition under restrictions]
\label{lem:restricted-decomposition}
Fix a restriction $\rho$, with set of free coordinates $\Gamma$.
For every $C\subseteq\Gamma$, define
$$\ket{v_C^\rho} = \sum_{\substack{A\subseteq\zo^n\\A\cap\Gamma=C}} \varepsilon_A(\rho)\ket{v_A}.$$
Then
$$\ket{\psi_x^\rho} = \sum_{\substack{C\subseteq\Gamma\\|C|\le d}} \chi_C(x)\ket{v_C^\rho}.$$
Moreover, the vectors $\{\ket{v_C^\rho}\}_C$ are pairwise orthogonal and satisfy $\sum_C\|v_C^\rho\|_2^2=1.$
\end{lemma}

\begin{proof}
Restricting the decomposition of $\ket{\psi_x}$ gives
$$\ket{\psi_x^\rho} = \sum_A \varepsilon_A(\rho) \chi_{A\cap\Gamma}(x) \ket{v_A}.$$
Grouping the terms according to $C=A\cap\Gamma$ yields
$$\ket{\psi_x^\rho} = \sum_C\chi_C(x)\ket{v_C^\rho}.$$

If $C\neq C'$, the vectors $\ket{v_C^\rho}$ and
$\ket{v_{C'}^\rho}$ are linear combinations of disjoint families
of the pairwise orthogonal vectors $\{\ket{v_A}\}_A$. Hence they
are orthogonal. Furthermore,
$$\sum_C\|v_C^\rho\|_2^2 = \sum_C \sum_{\substack{A\subseteq\zo^n\\A\cap\Gamma=C}} \|v_A\|_2^2 = \sum_A\|v_A\|_2^2 =1.$$
Finally, $C=A\cap\Gamma$ implies $|C|\le |A|\le d$.
\end{proof}

\subsection{Accepting-set Fourier growth bound and optimality of the algorithm}
\begin{theorem}[Accepting-set Fourier growth bound]
\label{thm:accepting-space}
Let $p : \{-1,1\}^{\zo^n} \to [0,1]$ be the acceptance probability of a $d$-query $\IQP$ circuit with accepting set $F$.
For every restriction $\rho$, $p|_\rho$ has degree at most $2d$ and
$$\|\widehat{p|_\rho}\|_1 \le \sqrt{|F|}.$$
Consequently, for every $0\leq\ell\leq 2d$, $L_{1,\ell}(p|_\rho)\le \sqrt{|F|},$ and $L_{1,\ell}(p|_\rho)=0$ for $\ell>2d$.
\end{theorem}

\begin{proof}
Fix a restriction $\rho$ and use the decomposition
$$\ket{\psi_x^\rho} = \sum_C\chi_C(x)\ket{v_C^\rho},$$
where the vectors $\ket{v_C^\rho}$ are pairwise orthogonal and satisfy
$\sum_C \|v_C^\rho\|_2^2=1.$
Let $M=H^{\otimes m}\Pi_FH^{\otimes m}.$
Then
$$p|_\rho(x) = \sum_{C,D} \chi_{C\triangle D}(x) \bra{v_C^\rho}M\ket{v_D^\rho}.$$
Since $|C|,|D|\le d$, this shows that $p|_\rho$ has degree at most $2d$.

By the dual characterization of the Fourier $\ell_1$ norm,
$$\|\widehat{p|_\rho}\|_1=\sup_{\alpha_S\in\{\pm1\}}\left|\sum_S \alpha_S\widehat{p|_\rho}(S)\right|.$$
For a choice of signs $\alpha=(\alpha_S)_S$, define
$$T^{\rho,\alpha}=\sum_{C,D}\alpha_{C\triangle D}\ket{v_D^\rho}\bra{v_C^\rho}.$$
Then
$$\sum_S\alpha_S\widehat{p|_\rho}(S)=\operatorname{Tr}(MT^{\rho,\alpha}).$$
Hence
$$\left|\operatorname{Tr}(MT^{\rho,\alpha})\right|\le\|M\|_F\|T^{\rho,\alpha}\|_F.$$

Since the vectors $\ket{v_C^\rho}$ are pairwise orthogonal,
the rank-one operators
$\ket{v_D^\rho}\bra{v_C^\rho}$ are pairwise orthogonal in the
Hilbert--Schmidt inner product. Therefore
$$\|T^{\rho,\alpha}\|_F^2 = \sum_{C,D}\|v_C^\rho\|_2^2\|v_D^\rho\|_2^2 =\left(\sum_C\|v_C^\rho\|_2^2\right)^2 =1.$$
Finally, $M$ is an orthogonal projector of rank $|F|$, so $\|M\|_F=\sqrt{|F|}.$
Taking the supremum over $\alpha$ gives $\|\widehat{p|_\rho}\|_1\le\sqrt{|F|}.$
\end{proof}

It is shown in~\cite{BS21,Gir26} that if $\mathcal{F}$ is a family of $2N$-variate Boolean functions closed under restrictions, then the maximum advantage with which $\mathcal{F}$ solves $2$-\textsc{Forrelation} is at most
$$O\!\left(\frac{L_{1,2}(\mathcal{F})}{\sqrt{N}}\right).$$

\begin{corollary}[Tight threshold on $|F|$]
	\label{cor:F-tight}
	Any $\IQP$ circuit solving $2$-\textsc{Forrelation} with constant advantage must have $|F| = \Omega(N)$. Our construction in \cref{Section:3} achieves $|F| = N$ (odd) or $|F|=2N$ (even).
\end{corollary}

\subsection{Query dependant Fourier growth bounds}
We originally obtained the following bound through a direct analysis of the $\IQP$ Fourier decomposition. 
We subsequently noticed that it is already implicit in the non-adaptive case of Girish, Sinha, Tal, and Wu~\cite{GSTW24}.
We nevertheless include the proof, since in the present $\IQP$ setting it is short and direct, and it clarifies the dependence on $d$ and $N$.
Moreover, we do not perform the final geometric-mean relaxation used in~\cite{GSTW24}, which preserves a sharper asymmetric bound.

\begin{theorem}\label{thm:d-query-bound}
	Let $p : \{-1,1\}^{\zo^n} \to [0,1]$ be the acceptance probability of a $d$-query $\IQP$ circuit.
	For every restriction $\rho$, $p|_\rho$ is a multilinear polynomial of degree at most $2d$ and
	$$
	L_{1,\ell}(p_{|\rho}) \leq
	\begin{cases}
		\displaystyle
		\sum_{\substack{0\leq a,b\leq d\\ a+b=\ell}}
		\sqrt{
			\min\left\{
				\binom{d}{a}\binom{N}{b},
				\binom{d}{b}\binom{N}{a}
			\right\}
		},
		& \text{if } 0\leq \ell\leq 2d,\\[1.2em]
		0,
		& \text{if } \ell>2d.
	\end{cases}
	$$
\end{theorem}

\begin{proof}
	Fix a restriction $\rho$, and let
	$\Gamma\subseteq\zo^n$ be the set of variables left free by $\rho$.
	By Lemma~\ref{lem:restricted-decomposition}, the state before the
	final Hadamards can be written as
	$$\ket{\psi_x^\rho}=\sum_{A\in\mathcal{A}_d^\rho}\chi_A(x)\ket{v_A^\rho},\qquad\mathcal{A}_d^\rho=\{A\subseteq\Gamma:|A|\leq d\},$$
	where the vectors $\ket{v_A^\rho}$ are pairwise orthogonal and satisfy
	$\sum_{A\in\mathcal{A}_d^\rho}\|\ket{v_A^\rho}\|_2^2=1.$

	Let $M=H^{\otimes m}\Pi_F H^{\otimes m}.$
	The restricted acceptance probability is
	$$p_{|\rho}(x)=\bra{\psi_x^\rho}M\ket{\psi_x^\rho}=\sum_{A,B\in\mathcal{A}_d^\rho}\chi_{A\triangle B}(x)\bra{v_A^\rho}M\ket{v_B^\rho}.$$
	Hence $p_{|\rho}$ has degree at most $2d$, and its Fourier
	coefficients are
	$$\widehat{p_{|\rho}}(S)=\sum_{\substack{A,B\in\mathcal{A}_d^\rho\\A\triangle B=S}}\bra{v_A^\rho}M\ket{v_B^\rho}.$$

	For signs $\alpha_S\in\{\pm1\}$, define the signed level-$\ell$
	Fourier growth
	$$L_{1,\ell}^\alpha(p_{|\rho})=\sum_{|S|=\ell}\alpha_S\widehat{p_{|\rho}}(S).$$
	Then
	$$L_{1,\ell}^\alpha(p_{|\rho})=\sum_{\substack{A,B\in\mathcal{A}_d^\rho\\|A\triangle B|=\ell}}\alpha_{A\triangle B}\bra{v_A^\rho}M\ket{v_B^\rho}.$$

	For $a,b\in\mathbb{N}$, let
	$$R_{a,b}=\left\{(A,B)\in\mathcal{A}_d^\rho\times\mathcal{A}_d^\rho:|A\setminus B|=a,\ |B\setminus A|=b\right\},$$
	and define
	$$l_{a,b}^\alpha=\sum_{(A,B)\in R_{a,b}}\alpha_{A\triangle B}\bra{v_A^\rho}M\ket{v_B^\rho}.$$
	Since
	$$|A\triangle B|=|A\setminus B|+|B\setminus A|,$$
	we have
	$$L_{1,\ell}^\alpha(p_{|\rho})=\sum_{\substack{0\leq a,b\leq d\\a+b=\ell}}l_{a,b}^\alpha.$$
	By the triangle inequality and the dual characterization of the
	$\ell_1$ norm,
	$$L_{1,\ell}(p_{|\rho})\leq\sum_{\substack{0\leq a,b\leq d\\a+b=\ell}}\sup_\alpha |l_{a,b}^\alpha|.$$

	We now bound $|l_{a,b}^\alpha|$.
	For $A\in\mathcal{A}_d^\rho$, let $P_A^\rho = \sum_{\substack{u\in\zo^m\\ A(u)\cap\Gamma=A}} \ket{u}\!\bra{u}.$ be an orthogonal projector. These projectors are pairwise orthogonal and satisfy
	$$P_A^\rho\ket{v_B^\rho}=\delta_{A,B}\ket{v_B^\rho},\qquad\sum_{A\in\mathcal{A}_d^\rho}P_A^\rho=I.$$
	Define
	$$M_{A,B}^\rho=P_A^\rho M P_B^\rho.$$
	Then
	$$\bra{v_A^\rho}M\ket{v_B^\rho}=\bra{v_A^\rho}M_{A,B}^\rho\ket{v_B^\rho}.$$

	For fixed $a,b$, define
	$$\ket{w_A}=\sum_{\substack{B\in\mathcal{A}_d^\rho\\(A,B)\in R_{a,b}}}\alpha_{A\triangle B}M_{A,B}^\rho\ket{v_B^\rho}.$$
	Then
	$$l_{a,b}^\alpha=\sum_{A\in\mathcal{A}_d^\rho}\braket{v_A^\rho}{w_A}.$$
	By Cauchy--Schwarz and
	$\sum_A\|v_A^\rho\|_2^2=1$,
	$$|l_{a,b}^\alpha|\leq\left(\sum_{A\in\mathcal{A}_d^\rho}\|w_A\|_2^2\right)^{1/2}.$$

	Define
	$$\deg_L(A,a,b)=\left|\{B\in\mathcal{A}_d^\rho:(A,B)\in R_{a,b}\}\right|,$$
	$$\deg_R(B,a,b)=\left|\{A\in\mathcal{A}_d^\rho:(A,B)\in R_{a,b}\}\right|,$$
	and
	$$\Delta_L(a,b)=\max_{A\in\mathcal{A}_d^\rho}\deg_L(A,a,b),\qquad\Delta_R(a,b)=\max_{B\in\mathcal{A}_d^\rho}\deg_R(B,a,b).$$

	Applying Cauchy--Schwarz again,
	\begin{align*}
		\|\ket{w_A}\|_2^2
		&\leq\deg_L(A,a,b)\sum_{\substack{B\in\mathcal{A}_d^\rho\\(A,B)\in R_{a,b}}}\|M_{A,B}^\rho\ket{v_B^\rho}\|_2^2\\
		&\leq \Delta_L(a,b)\sum_{B\in\mathcal{A}_d^\rho}\|M_{A,B}^\rho\ket{v_B^\rho}\|_2^2.
	\end{align*}
	Hence
	$$\sum_{A\in\mathcal{A}_d^\rho}\|\ket{w_A}\|_2^2 \leq \Delta_L(a,b)\sum_{A,B\in\mathcal{A}_d^\rho}\|M_{A,B}^\rho\ket{v_B^\rho}\|_2^2.$$

	Since $M$ is an orthogonal projector, $\|M\|_{\mathrm{op}}=1$. Using the orthogonality of the projectors $P_A^\rho$, we obtain
	\begin{align*}
		\sum_{A,B\in\mathcal{A}_d^\rho} \|M_{A,B}^\rho\ket{v_B^\rho}\|_2^2
		&= \sum_{B\in\mathcal{A}_d^\rho} \sum_{A\in\mathcal{A}_d^\rho} \|P_A^\rho M\ket{v_B^\rho}\|_2^2\\
		&= \sum_{B\in\mathcal{A}_d^\rho} \|M\ket{v_B^\rho}\|_2^2\\
		&\leq \sum_{B\in\mathcal{A}_d^\rho} \|\ket{v_B^\rho}\|_2^2\\
		&=1.
	\end{align*}
	Therefore
	$$|l_{a,b}^\alpha|\leq\sqrt{\Delta_L(a,b)}.$$
	Exchanging the roles of $A$ and $B$ gives
	$$|l_{a,b}^\alpha|\leq\sqrt{\Delta_R(a,b)}.$$
	Thus
	$$L_{1,\ell}(p_{|\rho})\leq\sum_{\substack{0\leq a,b\leq d\\a+b=\ell}}\sqrt{\min\{\Delta_L(a,b),\Delta_R(a,b)\}}.$$

	It remains to compute the degrees.
	For a fixed $A$, a set $B\in R_{a,b}$ is obtained by removing $a$ elements of $A$ and adding $b$ elements of $\Gamma\setminus A$.
	Since $B\in\mathcal{A}_d^\rho$, we must also have $|A|-a+b\leq d$.
	Therefore
	$$\deg_L(A,a,b)=\mathbbm{1}_{\{|A|-a+b\leq d\}}\binom{|A|}{a}\binom{|\Gamma|-|A|}{b}.$$
	Since $|A|\leq d$ and $|\Gamma|\leq N$,
	$$\Delta_L(a,b) \leq \binom{d}{a}\binom{N}{b}.$$
	Exchanging the roles of $A$ and $B$ gives
	$$\Delta_R(a,b) \leq \binom{d}{b}\binom{N}{a}.$$
	Hence, for $0\leq\ell\leq2d$,
	$$
	L_{1,\ell}(p_{|\rho})
	\leq
	\sum_{\substack{0\leq a,b\leq d\\a+b=\ell}}
	\sqrt{
		\min\left\{
			\binom{d}{a}\binom{N}{b},
			\binom{d}{b}\binom{N}{a}
		\right\}
	}.
	$$
	For $\ell>2d$, the Fourier coefficients vanish because $p_{|\rho}$ has degree at most $2d$.
\end{proof}

\begin{corollary}
	Let $p : \{-1,1\}^{\zo^n} \to [0,1]$ be the acceptance probability of a $d$-query $\IQP$ circuit.
	For every restriction $\rho$, $p|_\rho$ is a multilinear polynomial of degree at most $2d$ and
	$$L_{1,\ell}(p_{|\rho}) = \left\{\begin{array}{ll}
		O_{\ell}(d^{\lceil l/2 \rceil /2}N^{\lfloor l/2 \rfloor /2})
		& \text{if } 0\leq \ell \leq 2d \\
		0 & \text{if } \ell > 2d.
	\end{array}\right.$$
	In particular,
	$$L_{1,2}(p|_\rho) = O(\sqrt{dN}), \quad L_{1,3}(p|_\rho) = O(d\sqrt{N}), \quad L_{1,6}(p|_\rho) = O(d^{3/2}N^{3/2}).$$
\end{corollary}

The two bounds capture complementary phenomena.
The accepting-set argument controls the total Fourier mass globally through the Frobenius norm of the accepting projector, whereas the query structure controls how this mass can be distributed across Fourier levels through the possible symmetric differences $A\triangle B$.

The results of~\cite{BS21,Gir26} imply that if $\mathcal{F}$ is a family of $3N$-variate Boolean functions closed under restrictions, then the maximum advantage with which $\mathcal{F}$ solves $3$-\textsc{Forrelation} is at most
$$O\left(\frac{L_{1,3}(\mathcal{F})}{N}+\frac{L_{1,6}(\mathcal{F})}{N^2}\right).$$

Combining the above Fourier growth bounds for $d$-query $\IQP$ circuits with the $3$-\textsc{Forrelation} advantage bound yields a direct proof that $3$-\textsc{Forrelation} cannot be solved by $\IQP$ circuits making $d=o(N^{1/3})$ queries with constant advantage.
In particular, this recovers the separation of $3$-\textsc{Forrelation} from $\IQP$ directly at the level of $\IQP$ Fourier growth, without using that $\IQP \subseteq \hBQP$ and Girish's Fourier growth bound on $\hBQP$ \cite{Gir26}.

More generally, for every fixed $k>3$, combining the Fourier growth bounds above with Theorems~3.2 and~3.4 of~\cite{BS21} shows that any $\IQP$ circuit solving $k$-\textsc{Forrelation} with constant advantage requires
$$d = \Omega_k\left(N^{1-2/k}\right)$$
queries.
	
	\paragraph{Acknowledgments.} 
	We acknowledge funding from the French PEPR integrated projects EPIQ (ANR-22-PETQ-007), PQ-TLS (ANR-22-PETQ-008) and HQI (ANR-22-PNCQ-0002) all part of plan France 2030.
	
	\bibliography{paperbib}

@inbook{Rot09,
   title={Quantum Algorithms to Solve the Hidden Shift Problem for Quadratics and for Functions of Large Gowers Norm},
   ISBN={9783642038167},
   ISSN={1611-3349},
   url={http://dx.doi.org/10.1007/978-3-642-03816-7_56},
   DOI={10.1007/978-3-642-03816-7_56},
   booktitle={Mathematical Foundations of Computer Science 2009},
   publisher={Springer Berlin Heidelberg},
   author={Rötteler, Martin},
   year={2009},
   pages={663–674}
}

@inproceedings{Aar10,
	title={{BQP} and the polynomial hierarchy},
	author={Aaronson, Scott},
	booktitle={Proceedings of the forty-second ACM symposium on Theory of computing},
	pages={141--150},
	year={2010}
}

@inproceedings{AA15,
	title={Forrelation: A problem that optimally separates quantum from classical computing},
	author={Aaronson, Scott and Ambainis, Andris},
	booktitle={Proceedings of the forty-seventh annual ACM symposium on Theory of computing},
	pages={307--316},
	year={2015}
}

@inproceedings{AC17,
	title={Complexity-theoretic foundations of quantum supremacy experiments},
	author={Aaronson, Scott and Chen, Lijie},
	booktitle={Proceedings of the 32nd Computational Complexity Conference},
	pages={1--67},
	year={2017}
}

@article{BCJ25,
	title={Instantaneous quantum polynomial-time sampling and verifiable quantum advantage: Stabilizer scheme and classical security},
	author={Bremner, Michael J and Cheng, Bin and Ji, Zhengfeng},
	journal={PRX Quantum},
	volume={6},
	number={2},
	pages={020315},
	year={2025},
	publisher={APS}
}

@article{BJS11,
	title={Classical simulation of commuting quantum computations implies collapse of the polynomial hierarchy},
	author={Bremner, Michael J and Jozsa, Richard and Shepherd, Dan J},
	journal={Proceedings of the Royal Society A: Mathematical, Physical and Engineering Sciences},
	volume={467},
	number={2126},
	pages={459--472},
	year={2011},
	publisher={The Royal Society}
}

@article{BMS16,
	title={Average-case complexity versus approximate simulation of commuting quantum computations},
	author={Bremner, Michael J and Montanaro, Ashley and Shepherd, Dan J},
	journal={Physical review letters},
	volume={117},
	number={8},
	pages={080501},
	year={2016},
	publisher={APS}
}

@inproceedings{BS21,
	title={$k$-forrelation optimally separates quantum and classical query complexity},
	author={Bansal, Nikhil and Sinha, Makrand},
	booktitle={Proceedings of the 53rd Annual ACM SIGACT Symposium on Theory of Computing},
	pages={1303--1316},
	year={2021}
}

@inproceedings{BV93,
	title={Quantum complexity theory},
	author={Bernstein, Ethan and Vazirani, Umesh},
	booktitle={Proceedings of the twenty-fifth annual ACM symposium on Theory of computing},
	pages={11--20},
	year={1993}
}

@article{DJ92,
	title={Rapid solution of problems by quantum computation},
	author={Deutsch, David and Jozsa, Richard},
	journal={Proceedings of the Royal Society of London. Series A: Mathematical and Physical Sciences},
	volume={439},
	number={1907},
	pages={553--558},
	year={1992},
	publisher={The Royal Society London}
}

@inproceedings{Gir26,
    author = {Girish, Uma},
    title = {Fourier Spectrum of Noisy Quantum Algorithms},
    year = {2026},
    isbn = {9798400725364},
    publisher = {Association for Computing Machinery},
    address = {New York, NY, USA},
    url = {https://doi.org/10.1145/3798129.3800750},
    doi = {10.1145/3798129.3800750},
    booktitle = {Proceedings of the 58th Annual ACM Symposium on Theory of Computing},
    pages = {305–313},
    numpages = {9},
    location = {Salt Lake City, UT, USA},
    series = {STOC '26}
}

@article{IRR+21,
	title={Tight bounds on the Fourier growth of bounded functions on the hypercube},
	author={Iyer, Siddharth and Rao, Anup and Reis, Victor and Rothvoss, Thomas and Yehudayoff, Amir},
	journal={arXiv preprint arXiv:2107.06309},
	year={2021}
}

@article{JM24,
	title={The space just above one clean qubit},
	author={Jacobs, Dale and Mehraban, Saeed},
	journal={arXiv preprint arXiv:2410.08051},
	year={2024}
}

@article{Kah23,
	title={Forging quantum data: classically defeating an {IQP}-based quantum test},
	author={Kahanamoku-Meyer, Gregory D},
	journal={Quantum},
	volume={7},
	pages={1107},
	year={2023},
	publisher={Verein zur F{\"o}rderung des Open Access Publizierens in den Quantenwissenschaften}
}

@article{KL98,
	title={Power of one bit of quantum information},
	author={Knill, Emanuel and Laflamme, Raymond},
	journal={Physical Review Letters},
	volume={81},
	number={25},
	pages={5672},
	year={1998},
	publisher={APS}
}

@article{PLS+24,
	title={Robust sparse {IQP} sampling in constant depth},
	author={Paletta, Louis and Leverrier, Anthony and Sarlette, Alain and Mirrahimi, Mazyar and Vuillot, Christophe},
	journal={Quantum},
	volume={8},
	pages={1337},
	year={2024},
	publisher={Verein zur F{\"o}rderung des Open Access Publizierens in den Quantenwissenschaften}
}

@article{RT22,
	title={Oracle separation of {BQP} and {PH}},
	author={Raz, Ran and Tal, Avishay},
	journal={ACM Journal of the ACM (JACM)},
	volume={69},
	number={4},
	pages={1--21},
	year={2022},
	publisher={ACM New York, NY}
}

@article{SB09,
	title={Temporally unstructured quantum computation},
	author={Shepherd, Dan and Bremner, Michael J},
	journal={Proceedings of the Royal Society A: Mathematical, Physical and Engineering Sciences},
	volume={465},
	number={2105},
	pages={1413--1439},
	year={2009},
	publisher={The Royal Society London}
}

@article{SJ08,
	title={Estimating Jones polynomials is a complete problem for one clean qubit},
	author={Shor, Peter W and Jordan, Stephen P},
	journal={Quantum Information \& Computation},
	volume={8},
	number={8},
	pages={681--714},
	year={2008},
	publisher={Rinton Press, Incorporated Paramus, NJ}
}

@article{Sim97,
	title={On the power of quantum computation},
	author={Simon, Daniel R},
	journal={SIAM journal on computing},
	volume={26},
	number={5},
	pages={1474--1483},
	year={1997},
	publisher={SIAM}
}

@inproceedings{SSW21,
	title={An optimal separation of randomized and quantum query complexity},
	author={Sherstov, Alexander A and Storozhenko, Andrey A and Wu, Pei},
	booktitle={Proceedings of the 53rd Annual ACM SIGACT Symposium on Theory of Computing},
	pages={1289--1302},
	year={2021}
}

@inproceedings{Tal20,
	title={Towards optimal separations between quantum and randomized query complexities},
	author={Tal, Avishay},
	booktitle={2020 IEEE 61st Annual Symposium on Foundations of Computer Science (FOCS)},
	pages={228--239},
	year={2020},
	organization={IEEE}
}

@inproceedings{GSTW24,
  title={The power of adaptivity in quantum query algorithms},
  author={Girish, Uma and Sinha, Makrand and Tal, Avishay and Wu, Kewen},
  booktitle={Proceedings of the 56th Annual ACM Symposium on Theory of Computing},
  pages={1488--1497},
  year={2024}
}
	\bibliographystyle{alpha}	
\end{document}